\documentclass[a4paper,11pt]{article}

\input{preamble.tex}

\begin{document}

\title{Kernelization of Counting Problems}

\author{
Daniel Lokshtanov\thanks{University of California, Santa Barbara, USA. \texttt{daniello@ucsb.edu}}
 \and Pranabendu Misra\thanks{Chennai Mathematical Institute (CMI), India. \texttt{pranabendu@cmi.ac.in}}
 \and Saket Saurabh\thanks{University of Bergen and The Institute of Mathematical Sciences, HBNI, Chennai, India. \texttt{saket@imsc.res.in}}
 \and Meirav Zehavi\thanks{Ben-Gurion University, Beersheba, Israel. \texttt{meiravze@bgu.ac.il}} 
}

\maketitle




\begin{abstract} 

We introduce a new framework for the analysis of preprocessing routines for parameterized counting problems. Existing frameworks that encapsulate parameterized counting problems permit the usage of exponential (rather than polynomial) time either explicitly or by implicitly reducing the counting problems to enumeration problems. Thus, our framework is the only one  in the spirit of classic kernelization (as well as lossy kernelization). Specifically, we define a compression of a counting problem $P$  into a counting problem $Q$ as a pair of polynomial-time procedures: $\mathsf{reduce}$ and $\mathsf{lift}$. Given an instance of $P$, $\mathsf{reduce}$ outputs an instance of $Q$ whose size is bounded by a function $f$ of the parameter, and given the number of solutions to the instance of $Q$, $\mathsf{lift}$ outputs the number of solutions to the instance of $P$. When $P=Q$, compression is termed kernelization, and when $f$ is polynomial, compression is termed polynomial compression. Our technical (and other conceptual) contributions can be classified into two categories:

\smallskip\noindent {\bf Upper Bounds.} We prove two theorems: {\em (i)} The {\sc \#Vertex Cover} problem parameterized by solution size admits a polynomial kernel; {\em (ii)} Every problem in the class of {\sc \#Planar $\cal F$-Deletion} problems parameterized by solution size admits a polynomial compression.

\smallskip\noindent{\bf Lower Bounds.} We introduce two new concepts of cross-compositions: EXACT-cross-composition and SUM-cross-composition. We prove that if a \#P-hard counting problem $P$ EXACT-cross-composes into a parameterized counting problem $Q$, then $Q$ does not admit a polynomial compression unless the polynomial hierarchy collapses. We conjecture that the same statement holds for SUM-cross-compositions. Then, we prove that: {\em (i)} {\sc \#Min $(s,t)$-Cut} parameterized by treewidth does not admit a polynomial compression unless the polynomial hierarchy collapses; {\em (ii)} {\sc \#Min $(s,t)$-Cut} parameterized by minimum cut size, {\sc \#Odd Cycle Transversal} parameterized by solution size, and {\sc \#Vertex Cover} parameterized by solution size minus maximum matching size,  do not admit polynomial compressions unless our conjecture is false.

\end{abstract}

\newpage


\section{Introduction}\label{sec:intro}

Preprocessing is an integral part of almost any application, ranging from lossless data compression to microarray data analysis for the classification of cancer types.  Therefore, {\em kernelization} (or, more generally, {\em compression}), the mathematical paradigm to analyze preprocessing procedures, is termed ``the lost continent of polynomial time''~\cite{fellows2006lost}. Formally, a decision problem $P$ admits a {\em compression} into a decision problem $Q$ if there exists a polynomial-time algorithm that, given an instance $(I,k)$ of $P$, translates it into an equivalent\footnote{That is, $(I,k)$ is a yes-instance if and only if $(I',k')$ is a yes-instance.} instance $(I',k')$ of $Q$ of size $f(k)$ for some computable function $f$ that depends only on $k$.  When $P=Q$, a compression is termed {\em kernelization}. It is known that  a (decidable) problem admits a kernel if and only if it is in {\em fixed-parameter tractable (FPT)}~\cite{cai1997advice}.\footnote{We refer to Section \ref{sec:prelims} for basic definitions in parameterized complexity and graph theory.} Thus, the most central question in kernelization is: Which problems admit compressions (or kernels) of size $f(k)$ where $f $ is polynomial in $k$, termed {\em polynomial compressions}?
Techniques to show upper bounds on (polynomial or other) kernel sizes have already emerged in the early 1990s~\cite{fomin2019kernelization}. On the other hand, Bodlaender et al.~\cite{DBLP:journals/jcss/BodlaenderDFH09} proved that, unless the polynomial hierarchy collapses, there exist problems that do not admit a polynomial compression (and, hence, neither a polynomial kernel). 

Due to the centrality and mathematical depth of compression/kernelization, the underlying framework has been extended to capture optimization problems, and, more generally, the computation of approximate (rather than only exact) solutions for optimization problems, by Lokshtanov et al.~\cite{DBLP:conf/stoc/LokshtanovPRS17} (building upon \cite{DBLP:journals/jcss/FellowsKRS18}). In particular, a compression of an optimization problem $P$ into an optimization problem $Q$ is a pair of polynomial-time procedures: $\mathsf{reduce}$ and $\mathsf{lift}$. Given an instance of $P$, $\mathsf{reduce}$ outputs an instance of $Q$ whose size is bounded by a function $f$ of the parameter, and given an optimal solution to the instance of $Q$, $\mathsf{lift}$ outputs an optimal solution to the instance of $P$. More generally, to encompass the computation of approximate solutions with a loss of factor $\alpha\geq 1$, given a $\beta$-approximate solution to the instance of $Q$, for any $\beta\geq 1$, $\mathsf{lift}$ must output an $\alpha\cdot\beta$-approximate solution to the instance of $P$. Since its introduction,  this notion of compression/kernelization (termed {\em lossy compression/kernelization}) has already found a wide range of applications; see, e.g., \cite{DBLP:conf/soda/Manurangsi19,DBLP:conf/esa/0001KW19,eiben2019lossy,krithika2016lossy,krithika2018revisiting,agarwal2019parameterized,van2020approximate,eiben2017lossy}
 for just a few illustrative examples.
 
In this paper, we introduce a new framework for the analysis of preprocessing routines for parameterized counting problems. Existing frameworks that encapsulate parameterized counting problems permit the usage of exponential (rather than polynomial) time either explicitly or by implicitly reducing counting problems to enumeration problems (see Section \ref{sec:comparison}). Thus, our framework is the only one  in the spirit of classic compression/kernelization in particular, and lossy compression/kernelization in general. Specifically, we define a compression of a counting problem $P$  into a counting problem $Q$ as a pair of polynomial-time procedures: $\mathsf{reduce}$ and $\mathsf{lift}$. Given an instance of $P$, $\mathsf{reduce}$ outputs an instance of $Q$ whose size is bounded by a function $f$ of the parameter, and given the number of solutions to the instance of $Q$, $\mathsf{lift}$ outputs the number of solutions to the instance of $P$. We demonstrate the depth of our framework by proofs of both positive and negative results (see Section \ref{sec:contribution}). In particular, in terms of conceptual contribution, in addition to the framework itself, we also introduce two new types of cross-compositions, termed EXACT- and SUM-cross-compositions, aiming to provide analogs to the classic OR- and AND-cross-compositions used to derive negative results for (classic) kernels.
 
Over the past two decades, the body of works on parameterized counting problems has grown quite rapidly (see, e.g., \cite{DBLP:conf/stoc/FockeR22,DBLP:conf/stoc/Curticapean21,DBLP:conf/soda/LokshtanovSZ21,DBLP:conf/focs/0002R21,DBLP:conf/soda/DellLM20,DBLP:conf/soda/CurticapeanLN18,DBLP:conf/stoc/CurticapeanDM17} for a few illustrative examples of recent developments). In both theory and practice, there are various scenarios where counting the number of solutions might be equally (or more) important than only detecting a single solution (if one exists)~\cite{DBLP:conf/iwpec/Curticapean18}. This includes, for example, the computation of graph motifs to observe certain phenomena in social and biological networks~\cite{milo2002network}, and determination of thermodynamic properties of discrete systems by partition functions~\cite{isihara2013statistical}. However, most natural counting problems are not known to (and unlikely to) admit polynomial-time algorithms: Beyond problems whose decision versions are NP-hard, there also exist numerous problems whose decision versions are solvable in polynomial time, but whose counting versions are unlikely to be (e.g, a prime example of such problems is the {\sc Maximum Matching} problem on bipartite graphs~\cite{valiant1979complexity1,valiant1979complexity}). Naturally, this makes the study of the parameterized complexity of counting problems very attractive.

\subsection{Related Frameworks}\label{sec:comparison}

Prior to our work, there existed three frameworks relevant to the analysis of preprocessing routines for parameterized counting problems. However, all of these three frameworks (explicitly or implicitly)  correspond  to computation in exponential (rather than  polynomial) time, as well as to either enumeration (rather than counting) or data reduction other than compression/kernelization. Thus, they serve purposes that are very different than what {\em compression/kernelization} of parameterized {\em counting} problems should be (though, of course, they are of interest on their own right). Below, we elaborate on each of these three frameworks.

Among the three  aforementioned frameworks, the one whose utility is most similar to ours was developed by Thurley~\cite{DBLP:conf/tamc/Thurley07}, yet, even this framework concerns, implicitly, enumeration and computation in exponential time  (indeed, it is referred to as a formalization of so-called enumeration compactors in \cite{DBLP:conf/isaac/0002ST18}, and as a reduction of counting to enumeration in \cite{DBLP:journals/jcss/GolovachKKL22}). Roughly speaking, the definition of Thurley~\cite{DBLP:conf/tamc/Thurley07} can be interpreted as follows when using two polynomial-time procedures (as we do), $\mathsf{reduce}$ and $\mathsf{lift}$. Here, given an instance of a {\em counting} problem $P$, $\mathsf{reduce}$ outputs an instance of an {\em enumeration} problem $Q$ whose size is bounded by a function $f$ of the parameter. We suppose that each solution to the instance of $Q$ corresponds to a set of solutions to the instance of $P$; then, the collection of sets of solutions to the instance $P$ corresponding to the different solutions to the instance of $Q$ should form a partition of the set of solutions to the instance of $P$. Accordingly,  given a particular {\em solution} $s$ to the instance of $Q$, $\mathsf{lift}$ outputs the number of solutions to the instance of $P$ that correspond to $s$. In particular, given an {\em enumeration} of the solutions to the instance of $Q$, by calling $\mathsf{lift}$ for each one of them, we can obtain (in exponential time, depending on the number of solutions) the number of solutions to the instance of $P$.

The second framework is explicitly designed for enumeration problems. Still, we briefly discuss it here, since it shares some similarity to the framework of Thurley~\cite{DBLP:conf/tamc/Thurley07}. This framework was introduced by Creignou et al.~\cite{creignou2017paradigms} and refined by  Golovach et al.~\cite{DBLP:journals/jcss/GolovachKKL22}. Roughly speaking, in its latter incarnation, we are also given two polynomial-time procedures, $\mathsf{reduce}$ and $\mathsf{lift}$. Here, given an instance of an {\em enumeration} problem $P$, $\mathsf{reduce}$ outputs an instance of an {\em enumeration} problem $Q$ whose size is bounded by a function $f$ of the parameter. Then, $\mathsf{lift}$ is defined similarly as before, except that now, given a particular solution $s$ to the instance of $Q$, it enumerates (either in polynomial time or with polynomial delay) the solutions to the instance of $P$ that correspond to $s$. Like before, to derive the number of solutions to the instance of $P$, it is required to spend exponential time. 

The third framework is designed specifically for counting, but it is less in the spirit of compression/kernelization, and, accordingly, it is termed {\em compaction}. Additionally and similarly to the two aforementioned frameworks, it corresponds to computation in exponential time.  This framework was introduced by Kim et al.~\cite{DBLP:conf/isaac/0002ST18} (and further surveyed in \cite{thilikos2021compactors}). Roughly speaking, here we consider a polynomial-time procedure $\mathsf{compactor}$ (that can be thought of as $\mathsf{reduce}$) and an {\em exponential-time (or worse)} procedure $\mathsf{extractor}$ (that is very different in spirit than $\mathsf{lift}$). Here, given an instance of a counting problem $P$, $\mathsf{compactor}$ outputs an instance of a counting problem $Q$ whose size is bounded by a function $f$ of the parameter. Having computed the output instance, one can essentially discard all knowledge of the input instance, yet call the procedure $\mathsf{extractor}$ to solve the input instance. In a sense, the definition of compaction can be viewed as an ``intermediate'' concept that lies in between those of a fixed-parameter algorithm and a compression algorithm, which is of interest on its own right. Perhaps the main drawback of this third framework is that, because $\mathsf{extractor}$ is allowed (and must be allowed, if we deal with a \#P-hard problem) to spend  exponential-time (or worse)  in the size of the output of $\mathsf{compactor}$, we might often want to employ, in the first place, a fixed-parameter algorithm directly on the instance of $P$. 

We note that we are not aware, with respect  to any of the three frameworks discussed above, of the establishment of any non-trivial lower bound---that is, a lower bound that does not simply follows from fixed-parameter intractability.

\medskip
\noindent{\bf Remark.} We were very recently made aware that independently of our work, Jansen and van der Steenhoven~\cite{newIPEC} have just presented results that are  more in-lined in spirit with ours: Specifically, they either solve the given instance, or output an instance of size polynomial in the parameter and with the same number of solutions. They also speculate on developing a meaningful theory of counting kernelization. We answer this speculation as in this paper, as we develop a framework counting kernelization, along with a framework for proving lower-bounds.

\subsection{Our Contribution}\label{sec:contribution}

Our technical (and other conceptual) contributions can be classified into two categories: upper bounds and lower bounds. Here, we discuss the statements our results, and the new concepts that we introduce in the context of lower bounds. The technical aspects of our work are overviewed later, in Section \ref{sec:overview}. (We remark that some additional simple statements concerning our notion of compression/kernelization are proved in Section \ref{sec:kernelDef}.)

\medskip
\noindent{\bf Upper Bounds.} Let us start with the discussion of our upper bounds. We begin by the analysis of the {\sc \#$k$-Vertex Cover} problem, whose decision version is the most well studied problem in parameterized complexity~\cite{DBLP:books/sp/CyganFKLMPPS15,downey2013fundamentals}. The objective is  to count the number of vertex covers of size at most $k$ in a given graph $G$. Here, it is important to note that we count {\em all} vertex covers of size at most $k$, and not only the minimal ones (which is a significantly easier task; see Section \ref{sec:vc}). For the {\sc \#$k$-Vertex Cover} problem, we prove the following theorem in Section \ref{sec:vc}.

\begin{restatable}{theorem}{vcKernel}
\label{thm:vcKernel}
{\sc \#$k$-Vertex Cover} admits a polynomial kernel.
\end{restatable}

Next, we turn to consider a wide  class of parameterized counting problems, termed the class of {\sc \#$k$-Planar ${\cal F}$-Deletion} problems. In particular, the class of {\sc $k$-Planar ${\cal F}$-Deletion} problems encompasses a wide variety of well-known problems that have been extensively studied from the viewpoint of parameterized complexity, such as {\sc Vertex Cover}, {\sc Feedback Vertex Set}, {\sc Treewidth $\eta$-Deletion}, and more~\cite{fomin2012planar}. While we present a meta-theorem that resolves every problem in this class, we do not generalize our previous theorem---our meta-theorem yields compressions rather than kernelizations. Formally, the class of {\sc \#$k$-Planar ${\cal F}$-Deletion} problems contains one problem for every (finite) set of connected graphs $\cal F$ that contains at least one planar graph---here, given a graph $G$ and $k\in\mathbb{N}_0$, the objective is to count the number of vertex sets of size at most $k$ whose removal from $G$ yields a graph that does not contain any graph from $\cal F$ as a minor. For the class of {\sc \#$k$-Planar ${\cal F}$-Deletion} problems, we prove the following theorem in Section \ref{sec:compression}.

\begin{restatable}{theorem}{compression}
\label{thm:compression}
{\sc \#$k$-Planar ${\cal F}$-Deletion} admits a polynomial compression. 
\end{restatable}

\medskip
\noindent{\bf Lower Bounds.} We present two new types of cross-compositions, which we term EXACT-cross-composition and SUM-cross-composition. To understand the roots of these notions, let us first briefly present the classic notion of OR-cross-composition. Roughly speaking, we say that a decision problem $P$ OR-cross-composes into a parameterized problem $Q$ if, given a set of instances $x_1,x_2,\ldots,x_t$ of $P$, we can, in polynomial time, output a single instance $(y,k)$ of $Q$ with the following properties: {\em (i)} the parameter $k$ is bounded by a polynomial function of $\max_{i=1}^t|x_i|$ and $\log t$, and {\em (ii)} $(y,k)$ is a yes-instance if and only if at least one $x_i$ is a yes-instance. The importance of the notion of OR-cross-composition to compression/kernelization is rooted at the following theorem: If an NP-hard problem $P$ OR-cross-composes into a parameterized problem $Q$, then, $Q$ does not admit a polynomial compression (and, hence, neither a polynomial kernel), unless coNP $\subseteq$ NP/poly~\cite{bodlaender2009problems,bodlaender2014kernelization}. The intuition behind the correctness of this theorem is that, if $Q$ did admit a polynomial compression, then that would have meant that, in polynomial time, we are able to turn $t$ instances of an NP-hard problem to a single instance whose size depends (roughly) only on that size of a polylogarithmic number of them rather than all of them---intuitively, this means that we were able to resolve instances of an NP-hard problem in polynomial time.

Now, let us first discuss our notion of EXACT-cross-composition.\footnote{In the manuscript, we consider SUM-cross-composition first since the reduction we give in the context of EXACT-cross-composition builds upon one of the reductions that we give in the context of SUM-cross-composition.} Roughly speaking, we say that a counting problem $P$ EXACT-cross-composes into a parameterized counting problem $Q$ if, given a set of instances $x_1,x_2,\ldots,x_t$ of $P$, we can, in polynomial time, output a single instance $(y,k)$ of $Q$ with the following properties: {\em (i)} the parameter $k$ is bounded by a polynomial function of $\max_{i=1}^t|x_i|$ and $\log t$, and {\em (ii)} given the number of solutions to $(y,k)$, we can output, in polynomial time, the number of solutions to $x_i$ for every $i\in\{1,2,\ldots,t\}$. For EXACT-cross-composition, we prove the following theorem in Section \ref{sec:lower}.

\begin{restatable}{theorem}{lowerThmEXACT}
\label{thm:exact}
Assume that a \#P-hard counting problem $P$ EXACT-cross-composes into a parameterized counting problem $Q$. Then, $Q$ does not admit a polynomial compression, 
 unless \#P $\subseteq$ ``NP/poly'' (which implies that coNP $\subseteq$ NP/poly).
\end{restatable}

For an application of Theorem \ref{thm:exact}, we consider the classic {\sc \#Min $(s,t)$-Cut} problem. Here, given a graph $G$ and two vertices $s,t$ in $G$, the objective is to count the number of minimum $(s,t)$-cuts in $G$. Notably, the decision version of this problem is solvable in polynomial time~\cite{cormen2001introduction} (and, hence, it trivially admits a polynomial, and even  constant-size, kernel, with respect to any parameter). Moreover, it is easy to see  that {\sc \#Min $(s,t)$-Cut} parameterized by treewidth is in FPT. So, it is natural to ask whether  {\sc \#Min $(s,t)$-Cut} parameterized by treewidth admits a polynomial kernel (or at least a polynomial compression). We answer this question negatively in Section \ref{sec:lower}. 

\begin{restatable}{theorem}{lowerEXACT}
\label{thm:lowerEXACT}
{\sc \#$w$-Min $(s,t)$-Cut} does not admit a polynomial compression, 
unless \#P $\subseteq$ ``NP/poly'' (which implies that coNP $\subseteq$ NP/poly).
\end{restatable}

Lastly, let us discuss our notion of SUM-cross-composition. Roughly speaking, we say that a counting problem $P$ SUM-cross-composes into a parameterized counting problem $Q$ if, given a set of instances $x_1,x_2,\ldots,x_t$ of $P$, we can, in polynomial time, output a single instance $(y,k)$ of $Q$ with the following properties: {\em (i)} the parameter $k$ is bounded by a polynomial function of $\max_{i=1}^t|x_i|$ and $\log t$, and {\em (ii)} the number of solutions to $(y,k)$ is equal to the sum of the number of solutions to $x_i$ over every $i\in\{1,2,\ldots,t\}$. For SUM-cross-composition, we have the following conjecture, termed the SUM-conjecture: If a \#P-hard counting problem $P$ SUM-cross-composes into a parameterized counting problem $Q$, then $Q$ does not admit a polynomial compression. The reason why we believe that this conjecture is true is rooted at the exact same intuition  mentioned earlier for the correctness of the corresponding theorem for OR-cross-composition.

As applications of our conjecture, we again consider the {\sc \#Min $(s,t)$-Cut} problem, now parameterized by the size of a minimum $(s,t)$-cut (which is in FPT~\cite{berge2019fixed}). Additionally, we consider the {\sc \#Odd Cycle Transversal} problem parameterized by solution size and the {\sc \#Vertex Cover} problem parameterized by solution size minus either its LP-value or the size of a maximum matching (we refer to Section \ref{sec:prelims} for formal definitions). We remark that the decision versions of these parameterized counting problems are known to admit polynomial kernels~\cite{kratsch2020representative}. For the aforementioned parameterized counting problems, we prove the following theorem in Section \ref{sec:lowerSum}.

\begin{restatable}{theorem}{lowerSUM}
\label{thm:lowerSUM}
{\sc \#$k$-Min $(s,t)$-Cut}, {\sc \#$k$-Odd Cycle Transversal}, {\sc \#$\ell$-Vertex Cover} and {\sc \#$m$-Vertex Cover}  do not admit polynomial compressions, 
unless the SUM-conjecture is false.
\end{restatable}


\section{Overview of Our Proofs}\label{sec:overview}

In what follows, we present an overview for the proofs of our main theorems.

\bigskip
\noindent{\bf Proof of Theorem \ref{thm:vcKernel}.} Our reduction consists of two steps. Here, we note that most of our efforts are invested in the second step. The first step yields two graphs: $G_1$ and $G_2$. We begin by an exhaustive application of the classic Buss rule (Definition \ref{BussRule}) on the input instance $(G,k)$. In particular, unless the answer is $0$, this yields an instance $(G_1,k_1)$ with $k_1\leq k$ and $|E(G_1)|\leq k_1^2$ whose number of solution equals the number of solutions to $G$. At this point, we do not have a kernel (or compression)---$|V(G_1)|$ can contain arbitrarily many isolated vertices. So, we define $G_2$ as $G_1$ where all isolated vertices are removed. However, the number of solutions (denoted by $x_2$) to $(G_2,k_1)$ can be very different than the number of solutions (denoted by $x_1$) to $(G_1,k_1)$, and it is unclear how to derive the second from the first. Specifically, suppose that $y_i$, $i\in\{1,2,\ldots,k_2\}$, is the number of solutions to $(G_2,k_1)$ of size exactly $i$. It is easy to see that $x_1=\sum_{i=0}^{k_2}(y_i\cdot\sum_{j=0}^{k_2-i}{|V(G_1)|-|V(G_2)|\choose j})$. However, by knowing $x_2$, we cannot know the individual values of the $y_i$'s! Although $x_2=\sum_{i=0}^{k_2}y_i$, there can be more than one choice (in fact, there can be exponentially many choices) for the $y_i$'s given only the knowledge of $x_2$.

Due to the above difficulty, we perform the second step of our reduction. Roughly speaking, we define $G_3$ (in Definition \ref{def:defineG3}) by the replacement of each vertex of $G_2$ by $d$ copies (false twins) of that vertex, and the addition of $t$ new isolated vertices. To make latter calculations work, we pick $d=|V(G_2)|\leq \OO((k_2)^2)$, and we pick $t$ to be ``large enough'' compared to $d$. Now, our main objective is to prove how from the number of solutions to $(G_3,k_3)$ (denoted by $x_3$), we can derive the individual values of the $y_i$'s.

To achieve the above-mentioned objective, we define a mapping from the set of solutions to $(G_2,k_1)$ to the power set of the set of solutions to $(G_3,k_3)$. Specifically, each vertex subset (in a collection denoted by $\mathsf{Map}(X)$) of $G_3$ that is mapped to a solution $X$ to $(G_2,k_1)$ is the union of all ``copies'' of each vertex in $X$ as well as at most $k_3-d\cdot|X|$ many other vertices from $G_3$ so that there does not exist a vertex outside $X$ having all of its copies chosen (Definition \ref{def:map}). We first assert that this mapping corresponds to a partition of the set of solutions to $(G_3,k_3)$ (Lemma \ref{lem:union}). Then, we turn to analyze the sizes of the mapped collections. Towards this, we begin with a simple proof that for every $X$ is of size $i$, for some $i\in\{0,1,\ldots,k_2\}$, the size of $|\mathsf{Map}(X)|$ is the same (denoted by $w_i$), captured by an explicit formula (Lemma \ref{lem:sizes}). In particular, $x_3=\sum_{i=0}^{k_2}y_i\cdot w_i$. Consider this equality as Equation (*).

The main property of the $w_i$'s is that, for every $i\in\{0,1,\ldots,k_2\}$, $w_i$ is ``significantly'' larger than the sum of all $w_j$'s for $j<i$ (proved in Lemma \ref{lem:bound}). In particular, based on Equation (*) and this property, we can derive, from $x_3$, the individual values of the $y_i$'s. Specifically, this can be done by the following loop. For $i=0,1,2,\ldots,k_2$, we let $y_i\leftarrow \lfloor x_3/w_i\rfloor$, and update $x_3\leftarrow x_3-y_i\cdot w_i$. This computation can be performed efficiently (in polynomial time), since the $w_i$'s can be computed efficiently by dynamic programming (Lemma \ref{lem:computewi}).
 In turn, this computation is the main part of the procedure $\mathsf{lift}$, presented in Section \ref{sec:vcLift}.

\bigskip
\noindent{\bf Proof of Theorem \ref{thm:compression}.}
At a high level, we follow the approach of ~\cite{fomin2012planar} who give a polynomial kernel for {\sc Planar-$\cal F$ Deletion}. Given an instance $(G,k)$, we compute a modulator $X$ using an approximation algorithm~\cite{fomin2012planar} (see Proposition~\ref{prop:pfd-modulator}). This modulator has size $k^{\OO(1)}$, assuming that $G$ has a $\cal F$-deletion set of size at most $k$. Next, we consider the components of $G - X$. A component $C$ is \emph{irrelevant}, if it is disjoint from every minimal $\cal F$-deletion set of size at most $k$. Using the properties of $\cal F$-free graphs (see Proposition~\ref{prop:PFD-tw} and Proposition~\ref{prop:pfd-modulator}), we obtain that all but $k^{\OO(1)}$ components of $G - X$ are irrelevant (see Lemma~\ref{lemma:pdf-irr-comp}). We delete all irrelevant components in the first phase of the reduction step. Let $G'$ be the resulting graph.

The next reduction step, considers each component of $G' - X$. For each such component $C$, we observe that it is a \emph{near-protrusion}~\cite{fomin2012planar}, i.e. a subgraph that has constant-treewidth and after the removal of a $\cal F$-deletion set from $G'$, has a constant sized boundary.
We then apply several powerful results on \emph{boundaried graphs}, summarized in Section~\ref{sec:compression-prelims} (also see ~\cite{fomin2019kernelization} for details), to show that the information required to count the number of $\cal F$-deletion sets of size $k'$ in $G'$, for every $k' \leq k$, can stored in a compressed form using $k^{\OO(1)}$ space.

Briefly, a boundaried graph is a graph $H$ where a subset of vertices $B$ are marked as boundary vertices. These boundary vertices are labeled with integers. Given two boundaried graphs $H_1$ and $H_2$, whose boundary vertices are labeled using the same set of integers, we can ``glue'' them to obtain a graph $H_1 \oplus H_2$, which is obtained by first taking a disjoint union of the two graphs and then identifying boundary vertices with the same label. Using the notion of boundaried graphs and gluing, we can define an equivalence relation, $\equiv_{\cal F}$ such that $H_1 \equiv_{\cal F} H_2$ if and only if for any other boundaried graph $H_3$, $H_1 \oplus H_3$ is $\cal F$-minor free $\iff$ $H_2 \oplus H_3$ is $\cal F$-minor free. It is known that this equivalence relation has finitely many equivalence class for any fixed $\cal F$ (see Proposition~\ref{prop:pdf-fii}).

Intuitively, our compression for a connected component $C$ of $G'-X$, considers the effect of deleting a $\cal F$-deletion set $S$ from $G'$, and records the number of ways this can happen. Since $C$ is a near protrusion, it has constant-treewidth and a constant size boundary in $G-S$ that is a subset of $X \setminus S$. We treat $G[(V(C) \cup N(C)) \setminus S]$ as a boundaried graph with boundary $N(C) \setminus S$, and note that $N(C) \subseteq X$. Note that $G[(V(C) \cup N(C)) \setminus S]$ lies in an equivalence class $\cal R$ of $\equiv_{\cal F}$.
Then, for each choice of $\cal R$, $N[C] \setminus S$ and $|S \cap V(C)|$ we record the number of subsets $S_C$ of $V(C)$ such that $G[(V(C) \cup N(C)) \setminus S]$ with boundary $N[C] \setminus S$ forms a boundaried graph that lies in $\cal R$. We compute and store this information in a table $T_C$ for each component $C$. We show that the number of such choices is bounded by $k^{\OO(1)}$, and each entry of $T_C$ can be computed polynomial time. We then argue that the information stored in the table is sufficient to compute $\mathsf{count}(k')$ which is the number of $\cal F$-deletion sets in $G'$ of size at most $k'$, for every $k' \leq k$.
Note that computing $\mathsf{count}(k')$ takes time exponential in $k$.
See Section~\ref{sec:compression-phase-2} for details.

The output of the {\sf reduce} procedure for {\sc \#Planar-$\cal F$ Deletion}, given an instance $(G,k)$, is a modulator $X$ of size $k^{\OO(1)}$ and a collection of tables $\{T_C\}$, one for each non-irrelevant component of $G - X$. Note that the size of the output is $k^{\OO(1)}$.
Next, the {\sf lift} procedure is given the instance $(G,k)$, the modulator $X$, the collection of tables $T_C$ for each component of $G'-X$, and finally the values $\{\mathsf{count}(k') \mid k' \leq k\}$.
The {\sf lift} procedure first computes $\tau_{irr}$ which denotes the total number of vertices in the irrelevant components of $G-X$. Then, from $\{\mathsf{count}(k') \mid k' \leq k\}$ and $\tau_{irr}$ it is easy to count the total number of solutions of size at most $k$ in $G$ in polynomial time.
The {\sf reduce} and {\sf lift} procedures together prove this theorem.
We refer to Section~\ref{sec:compression} for details.

\bigskip
\noindent{\bf Proof of Theorem \ref{thm:lowerSUM}.} We start with the proof that {\sc \#Min $(s,t)$-Cut} (which is \#P-hard~\cite{provan1983complexity}) SUM-cross-composes into {\sc \#$k$-Min $(s,t)$-Cut} (Lemma \ref{lem:stCutSum}). Suppose that we are given $\ell$ instances of {\sc \#Min $(s,t)$-Cut}, $(G_1,s_1,t_1),(G_2,s_2,t_2),\ldots,(G_\ell,s_\ell,t_\ell)$, where the size of a minimum $(s_i,t_i)$-cut in $G_i$ is assumed to be equal to the size of a minimum $(s_j,t_j)$-cut in $G_j$, for every $i,j\in[\ell]$. (This assumption is justified by the more general definition of cross-compositions that makes use of equivalence relations.) Then, the reduction to a single instance $(G,s,t)$ of is performed as follows: We take the disjoint union of the input graphs, and unify $t_i$ with $s_{i+1}$, for every $i\in\{1,2,\ldots,\ell-1\}$; additionally, we let $s=s_1$ and $t=t_\ell$. With this construction at hand, it is easy to see that each minimum $(s,t)$-cut in $G$ corresponds to a minimum $(s_i,t_i)$-cut in one of the $G_i$'s, and vice versa.
Thus, we derive that the number of minimum $(s,t)$-cuts in $G$ equals the sum of the number of minimum $(s_i,t_i)$-cuts in $G_i$, over every $i\in\{1,2,\ldots,\ell\}$. Moreover, the parameter $k$ is trivially bounded from above by $\max_{i=1}^\ell|E(G_i)|$.

Having asserted that {\sc \#$k$-Min $(s,t)$-Cut} does not admit a polynomial compression under the SUM-conjecture, we transfer its hardness to the {\sc \#$k$-Odd Cycle Transversal} problem by the design of a {\em polynomial parameter transformation} (Definition \ref{def:PPTcounting}) in Lemma \ref{lem:PPTOCT}. Suppose that we are given an instance $(G,s,t)$ of {\sc \#$k$-Min $(s,t)$-Cut} where $G$ is a connected graph. Then, we first turn $G$ into a graph $G_1$ be subdividing each edge once. In particular, we thus derive that all paths in $G_1$ between vertices that correspond to vertices (rather than edges) in $G$ are of even length. Next, we turn $G_1$ into a graph $G_2$ by replacing each vertex of $G_1$ that corresponds to a vertex of $G$ by $k+1$ copies (false twins). Intuitively, this will have the effect that no minimal solution being of size at most $k$ to our instance of {\sc \#$k$-Odd Cycle Transversal} (defined immediately) will pick any vertex in $G_2$ that corresponds to a vertex in $G$ (since we deal with edge-cuts, this property must be asserted for our proof of correctness). Complementary to this, we will (implicitly) prove that our instance has no solution of size smaller than $k$, so every solution of size at most $k$ is of size exactly $k$ and a minimal one. The last step of the reduction is to turn $G_2$ into a graph $G'$ by adding two new adjacent vertices, $x_i$ and $y_i$, for every $i\in\{1,2,\ldots,k+1\}$, and making all the $x_i$'s adjacent to all the copies of $s$, and all the $y_i$'s adjacent to all the copies of $t$. With this construction of $G'$ at hand (and keeping the parameter $k$ unchanged), we are able to prove that: {\em (i)} every odd cycle in $G'$ contains at least one path from a copy of $s$ to a copy of $t$ that corresponds to an $(s,t)$-path in $G$, and {\em (ii)} every $(s,t)$-path in $G$ can be translated to some particular set of odd cycles in $G'$ such that, to hit that set with at most $k$ vertices, it only ``makes sense'' to pick vertices in $G'$ that correspond to edges in $G$. From this, we are able to derive that the number of minimum $(s,t)$-cuts in $G$ equals the number of odd cycles transversal of $G'$ of size at most $k$.

Lastly, having asserted that {\sc \#$k$-Odd Cycle Transversal} does not admit a polynomial compression under the SUM-conjecture, we transfer its hardness to the {\sc \#$\ell$-Vertex Cover} problem (where the parameter is $k$ minus the LP-value) and the {\sc \#$m$-Vertex Cover} problem (where the parameter is $k$ minus the maximum size of a matching) by the design of another polynomial parameter transformation. We remark that, since it always holds that $m\geq\ell$, the hardness for  {\sc \#$m$-Vertex Cover} implies the hardness for {\sc \#$\ell$-Vertex Cover}. While the transformation itself is the same as the known reduction from {\sc $k$-Odd Cycle Transversal} to {\sc $m$-Vertex Cover}  (Lemma 3.10 in \cite{DBLP:books/sp/CyganFKLMPPS15}), the analysis somewhat differs. In particular, for the correctness, we actually cannot use {\sc \#$k$-Odd Cycle Transversal} as the source problem, but only restricted instances of it, where for every odd cycle transversal $S$ of size at most $k$, the removal of $S$ from the input graph $G$ yields a connected graph. Then, we are able to show that the number of odd cycle transversals of $G$ of size at most $k$ is exactly half the number of vertex covers of the output graph $G'$ of size at most $k'$. (The parameter of the output instance, $k'-m$, equals $k$.)

\bigskip
\noindent{\bf Proof of Theorems \ref{thm:exact} and \ref{thm:lowerEXACT}.} The proof of Theorem \ref{thm:exact} follows the lines of, yet is not identical to, the proof of the analogous statement for OR-cross-composition (see Appendix \ref{sec:exactProof}). For example, one notable difference concerns the part of the proof where we need to define a problem whose solution is a function of solutions of another problem. While for OR-cross-compositions, the chosen function is the logical OR of the given solutions, for us the chosen function is a weighted summation of the given solutions with weights chosen so that, from the weighted sum, we can derive each individual solution (that is similar to the spirit of the $\mathsf{lift}$ procedure given as part of the proof of Theorem \ref{thm:vcKernel}).

For the proof of Theorem \ref{thm:lowerEXACT}, we prove that {\sc \#Min $(s,t)$-Cut} EXACT-cross-composes into {\sc \#$w$-Min $(s,t)$-Cut} (Lemma \ref{lem:EXACTdistillation}). The reduction begins by taking the instance $(G,s,t)$ built in the proof of the SUM-cross-composition discussed earlier. We note that the treewidth of $G$ equals the maximum treewidth of $G_i$, over every $i\in\{1,2,\ldots,\ell\}$. However, recall that this construction only yields that the number of solutions to $(G,s,t)$ (say, $q$) equals $\sum_{i=1}^\ell q_i$ where $q_i$ is the number of solutions to $(G_i,s_i,t_i)$. So, by knowing only $q$, we are not able to derive the individual $q_i$'s (there can be exponentially many options for their values). So, we further modify the graph $G$ to obtain a graph $G'$ as follows.  For the copy of each $G_i$ in $G$, we add a $2^{m\cdot(\ell-1)}$ internally vertex-disjoint paths from $s_i$ to $t_i$, where $m=2\max_{j=1}^\ell|E(G_j)|$: $m(i-1)$ of these paths have three internal vertices, and the rest have one internal vertex. Notice that, to separate $s_i$ and $t_i$ in this ``extended'' copy of $G_i$, we need to pick at least one edge from each of the newly added paths, and we have two (resp., four) options for which edge to pick from each of the paths with one (resp., three) internal vertices. Having this insight in mind, we are able to show that the  number of solutions to $(G',s,t)$ (say, $q'$) equals $\sum_{i=1}^\ell q_i\cdot 2^{m(i-1)+m(\ell-1)}$. In particular, the coefficient of each $q_i$ is ``significantly'' larger than the sum of the coefficients of all $q_j$, $j<i$. In turn, this allows us to derive, from $q'$, the individual values of the $q_i$'s (similarly, in this part, to the corresponding parts of the proofs of Theorem \ref{thm:vcKernel} and \ref{thm:exact}). Further, we show that the addition of the aforementioned paths does not increase the treewidth of the graph (unless it was smaller than $2$).


\section{Preliminaries}\label{sec:prelims}

Let $\mathbb{N}_0=\mathbb{N}\cup\{0\}$. For $n\in\mathbb{N}$, let $[n]=\{1,2,\ldots,n\}$. Given a universe $U$, let $2^U=\{A: A\subseteq U\}$.

\bigskip
\noindent{\bf Graph Notation.} Throughout the paper, we consider finite, simple, undirected graphs. Given a graph $G$, let $V(G)$ and $E(G)$ denote its vertex set and edge set, respectively. Given a subset $U\subseteq E(G)$, let $G-U$ denote the graph on vertex set $V(G)$ and edge set $E(G)\setminus U$. We say that $S\subseteq V(G)$ {\em covers} $U\subseteq E(G)$ if for every edge $\{u,v\}\in U$, $S\cap\{u,v\}\neq\emptyset$. A {\em vertex cover} of $G$ is a subset $S\subseteq V(G)$ that covers $E(G)$. The set $S$ is said to be {\em minimal} if every subset of it is not a vertex cover of $G$. An {\em independent set} of $G$ is a subset $S\subseteq V(G)$ such that $E(G)\cap\{\{u,v\}: v\in S\}=\emptyset$.  A {\em matching} in $G$ is a subset $S\subseteq E(G)$ such that no two edges in $S$ share an endpoint. Let $\mu(G)$ denote the maximum size of a matching in $G$. Given two distinct vertices $s,t\in V(G)$, an {\em $(s,t)$-cut in $G$} is a subset $S\subseteq E(G)$ such that in $G-S$, the vertices $s$ and $t$ belong to different connected components. An $(s,t)$-cut in $G$ is {\em minimum} if there does not exist an $(s,t)$-cut in $G$ of smaller size. Given a subset $U\subseteq V(G)$, let $G[U]$ denote the subgraph of $G$ induced by $U$, and let $G-U$ denote $G[V(G)\setminus U]$. An {\em odd cycle transversal of $G$} is a subset $S\subseteq V(G)$ such that $G-S$ does not contain any odd cycle (i.e., a cycle with an odd number of vertices, or, equivalently, of edges). The subdivision of an edge $\{u,v\}\in E(G)$ is the operation that removes $\{u,v\}$ from $G$, adds a new vertex $x$ to $G$, and adds the edges $\{u,x\}$ and $\{v,x\}$ to $G$. Given a graph $H$, we write $H\subseteq G$ to indicate that $H$ is a subgraph of $G$. A graph $G$ is {\em bipartite} if there exists a partition $(X,Y)$ of $V(G)$ such that $E(G)\subseteq \{\{x,y\}: x\in X, y\in Y\}$, that is, $X$ and $Y$ are independent sets. Note that a graph $G$ is bipartite if and only if it does not contain any odd cycle~\cite{diestel2017extremal}. We say that a graph $H$ is a {\em minor} of a graph $G$ if there exists a series of vertex deletions, edge deletions and edge contractions in $G$ that yields $H$. We say that $G$ is a {\em planar graph} if it can be drawn on the Euclidean plane so that its edges can intersect only at their endpoints. 

Treewidth is a structural parameter indicating how much a graph resembles a tree. Formally:

\begin{definition}\label{def:treewidth}
A \emph{tree decomposition} of a graph $G$ is a pair ${\cal T}=(T,\beta)$ of a tree $T$
and $\beta:V(T) \rightarrow 2^{V(G)}$, such that
\begin{enumerate}
\itemsep0em 
\item\label{item:twedge} for any edge $\{x,y\} \in E(G)$ there exists a node $v \in V(T)$ such that $x,y \in \beta(v)$, and
\item\label{item:twconnected} for any vertex $x \in V(G)$, the subgraph of $T$ induced by the set $T_x = \{v\in V(T): x\in\beta(v)\}$ is a non-empty tree.
\end{enumerate}
The {\em width} of $(T,\beta)$ is $\max_{v\in V(T)}\{|\beta(v)|\}-1$. The {\em treewidth} of $G$, denoted by $\mathsf{tw}(G)$, is the minimum width over all tree decompositions of $G$.
\end{definition}

\bigskip
\noindent{\bf Problems and Counting Problems.} A {\em decision problem} (or {\em problem} for short) is a language $P\subseteq \Sigma^\star$. Here, $\Sigma$ is a finite alphabet, and, without loss of generality, we can assume that $\Sigma=\{0,1\}$. Often, some strings in $\Sigma^\star$ are ``irrelevant'' to $P$ (specifically, they clearly do not belong to $P$)---e.g., when $P$ concerns graphs and a given string does not encode a graph; so, the term {\em instance of $P$} is loosely used for strings that are relevant to $P$ in some such natural sense. An {\em algorithm for $P$} is a procedure that, given $x\in\Sigma^\star$, determines whether $x\in P$. We say that an instance $x$ of a problem $P$ is {\em equivalent} to an instance $x'$ of a problem $Q$ if: $x\in P$ if and only if $x'\in Q$.  A {\em counting problem} is a mapping $F$ from $\Sigma^\star$ to $\mathbb{N}_0$.  As before, the term {\em instance of $F$} is loosely used---while one still needs to define the mapping of ``irrelevant'' strings, the consideration of this mapping will be immaterial to us. An {\em algorithm for $F$} is a procedure that, given $x\in\Sigma^\star$, outputs $F(x)$.  A counting problem $F$  is a {\em counting version} of a problem $P$ if, for every $x\in\Sigma^\star$, $x\in P$ if and only if $F(x)\geq 1$. When we refer to ``the'' counting version of a problem $P$, we consider the counting version of $P$ whose choice (among all counting versions of $P$) is widely regarded the most natural one, and it is denoted by $\#P$.


\bigskip
\noindent{\bf Parameterized Complexity.} We start with the definition of a parameterized problem.

\begin{definition}[{\bf Parameterized Problem}]
 A {\em parameterized problem} is a language $P \subseteq \Sigma^\star \times\mathbb{N}_0$, where $\Sigma$ is a fixed, finite alphabet. For an instance $(x, k) \in \Sigma^\star \times\mathbb{N}_0$, $k$ is called the {\em parameter}.
 \end{definition}
 
 An {\em algorithm for $P$} is a procedure that, given $(x,k)\in\Sigma^\star\times\mathbb{N}_0$, determines whether $(x,k)\in P$. We say that $P$ is {\em fixed-parameter tractable (FPT)} if there exists an algorithm for $P$ that runs in time $f(k)\cdot|x|^{\OO(1)}$ where $f$ is some computable function of $k$. Such an algorithm is called a {\em fixed-parameter algorithm}. The main tool to assert that one problem is in FPT based on an already known membership of another problem in FPT is the design of a PPT, defined as follows.

\begin{definition}[{\bf PPT}]
Let $P, Q \subseteq \Sigma^\star \times\mathbb{N}_0$ be two parameterized problems. A polynomial-time algorithm $A$ is a {\em polynomial parameter transformation (PPT)} from $P$ to $Q$ if, given an instance $(x, k)$ of  $P$, $A$ outputs an equivalent instance $(x', k')$ of $Q$ (i.e., $(x, k) \in P$ if and only if $(x', k') \in Q$) such that $k' \leq p(k)$ for some polynomial function $p$.
\end{definition} 
 
 A companion notion of FPT is that of a compression or a kernelization, defined as follows.
 
 \begin{definition}[{\bf Compression and Kernelization}]
 Let $P$ and $Q$ be two parameterized problems. A {\em compression} (or {\em compression algorithm} for $P$ is a polynomial-time procedure that, given an instance $(x,k)$ of $P$, outputs an equivalent instance $(x',k')$ of $Q$ where $|x'|,k'\leq f(k)$ for some computable function $f$. Then, we say that $P$ admits a {\em compression of size $f(k)$}. When $f$ is polynomial, then we say that $P$ admits a {\em polynomial compression}. Further, when $P=Q$, we refer to compression also as {\em kernelization}.
 \end{definition}
 
Now, we state two central propositions that concern kernelization.
 
 \begin{proposition}[\cite{cai1997advice}]\label{prop:FPTequivKer}
 Let $P$ be a parameterized problem that is decidable. Then, $P$ is FPT if and only if it admits a kernel.
 \end{proposition} 
 
 \begin{proposition}[Folklore; See, e.g., Theorem 15.15 in \cite{DBLP:books/sp/CyganFKLMPPS15}]\label{prop:PPT}
 Let $P, Q$ be two parameterized problems such that there exists a PPT from $P$ to $Q$. If $Q$ admits a polynomial compression, then $P$ admits a polynomial compression. 
 \end{proposition}
 
Towards the statement of the main tool to refute the existence of polynomial compressions (and, hence, also polynomial kernels) for specific problems, we state the two following definitions.
 
 \begin{definition}[{\bf Polynomial Equivalence Relation}]
An equivalence relation $R$ on a set $\Sigma^\star$ is a {\em polynomial equivalence relation} if the following conditions are satisfied:
\begin{itemize}
\item There exists an algorithm that, given strings $x, y \in \Sigma^\star$, resolves whether $x \equiv_R y$ in time polynomial in $|x| + |y|$.
\item The relation $R$ restricted to the set $\Sigma^{\leq n}$ has at most $p(n)$ equivalence classes, for some polynomial function $p$.
\end{itemize}
 \end{definition}

\begin{definition}[{\bf OR-Cross-Composition}]
Let $P \subseteq \Sigma^\star$ be a problem and $Q \subseteq \Sigma^\star \times\mathbb{N}_0$ be a parameterized problem. We say that $P$ {\em OR-cross-composes} into $Q$ if there exists a polynomial equivalence relation $R$ and an algorithm $A$, called an {\em OR-cross-composition}, satisfying the following conditions. The algorithm $A$ takes as input a sequence of strings $x_1,x_2,\ldots,x_t \in\Sigma^\star$ that are equivalent with respect to $R$, runs in time polynomial in $\sum_{i=1}^t |x_i|$, and outputs one instance $(y, k)\in \Sigma^\star \times\mathbb{N}_0$ such that:
\begin{itemize}
\item  $k \leq p(\max^t_{i=1} |x_i| + \log t)$ for some polynomial function $p$, and
\item $(y,k) \in Q$ if and only if there exists at least one index $i\in [t]$ such that $x_i \in P$.
\end{itemize}
\end{definition}

Now, we state the main tool to refute the existence of polynomial compressions for specific problems.

\begin{proposition}[\cite{bodlaender2009problems,bodlaender2014kernelization}]\label{prop:crossComp}
Assume that an NP-hard problem $P$ OR-cross-composes into a parameterized problem $Q$. Then, $Q$ does not admit a polynomial compression, unless coNP $\subseteq$ NP/poly.
\end{proposition}

We remark that an analogous proposition, where OR is replaced by AND, has been proved in \cite{drucker2015new}.

We proceed to the definition of a parameterized counting problem.

\begin{definition}[{\bf Parameterized Counting Problem}]
A {\em parameterized counting problem} is a mapping $F$ from $\Sigma^\star\times\mathbb{N}_0$ to $\mathbb{N}_0$. 
\end{definition}

An {\em algorithm for $F$} is a procedure that, given $(x,k)\in\Sigma^\star\times\mathbb{N}_0$, outputs $F(x,k)$. As before, we say that $P$ is {\em fixed-parameter tractable (FPT)} if there exists an algorithm for $P$ that runs in time $f(k)\cdot|x|^{\OO(1)}$ where $f$ is some computable function of $k$. A parameterized counting problem $F$  is a {\em counting version} of a parameterized problem $P$ if, for every $(x,k)\in\Sigma^\star\times\mathbb{N}_0$, $(x,k)\in P$ if and only if $F(x,k)\geq 1$. When we refer to ``the'' counting version of a parameterized problem $P$, we consider the counting version of $P$ whose choice (among all counting versions of $L$) is widely regarded the most natural one, and it is denoted by $\#P$.

\begin{definition}[{\bf PPT (Counting Version)}]\label{def:PPTcounting}
Let $P, Q: \Sigma^\star \times\mathbb{N}_0\rightarrow\mathbb{N}_0$ be two parameterized counting problems. A pair of polynomial-time procedures $(\mathsf{reduce},\mathsf{lift})$ is a {\em polynomial parameter transformation (PPT)} from $P$ to $Q$ such that:
\begin{itemize}
\item  Given an instance $(x, k)$ of  $P$, $\mathsf{reduce}$ outputs an instance $(x', k')$ of $Q$ such that $k' \leq p(k)$ for some polynomial function $p$.
\item Given an instance $(I,k)$ of $P$, the instance $(I',k')$ that is the output of $\mathsf{reduce}$ on $(I,k)$, and $x'$ such that $P(I',k')=x'$, $\mathsf{lift}$ outputs $x$ such that $Q(I,k)=x$.
\end{itemize}
\end{definition}

 The main concept to show that a problem is unlikely to be FPT is the one of parameterized reductions analogous to those employed in classical complexity. Here, the concept of W[1]-hardness replaces the one of NP-hardness, and for reductions we need not only construct an equivalent instance in FPT time, but also ensure that the size of the parameter in the new instance depends only on the size of the parameter in the original one. If there exists such a reduction transforming a parameterized problem known to be W[1]-hard to another parameterized problem $P$, then the problem $P$ is W[1]-hard as well. Central W[1]-hard problems include, for example, deciding whether a nondeterministic single-tape Turing machine accepts within $k$ steps, {\sc Clique} parameterized be solution size, and {\sc Independent Set} parameterized by solution size. Naturally, \#W[1]-hardness is the concept analogous to W[1]-hardness in the realm of parameterized counting problems. For more information on W[1]-hardness and \#W[1]-hardness, we refer to \cite{DBLP:books/sp/CyganFKLMPPS15,DBLP:conf/iwpec/Curticapean18,downey2013fundamentals}.

\bigskip
\noindent{\bf Problem Definitions.} The counting problems studied in this paper are defined as follows.
\begin{itemize}
\item {\sc \#$k$-Vertex Cover} ({\sc \#$k$-Minimal Vertex Cover}): Given a graph $G$ and a non-negative integer $k$, output the number of vertex covers (minimal vertex covers) of $G$ of size at most $k$. Here, the parameter is $k$.
\item {\sc \#$\ell$-Vertex Cover} and {\sc \#$m$-Vertex Cover}:  Defined as {\sc \#$k$-Vertex Cover} with the exception that the parameters $\ell$ and $m$ are $k-\mathsf{LP_{VC}}(G)$ and $k-\mu(G)$, respectively. Here, $\mathsf{LP_{VC}}(G)$ denotes the optimum of the (standard) linear program that corresponds to {\sc Vertex Cover} (see \cite{DBLP:books/sp/CyganFKLMPPS15}, Section 3.4).
\item {\sc \#$k$-Planar ${\cal F}$-Deletion}: Let ${\cal F}$ be a finite set of connected graphs that contains at least one planar graph. Given a graph $G$ and a non-negative integer $k$, output the number of subsets $S\subseteq V(G)$ of size at most $k$ such that $G-S$ does not contain any graph from $\cal F$ as a minor. We remark that the {\sc \#$k$-Planar ${\cal F}$-Deletion} problem encompasses (based on different choices of $\cal F$) various other problems, such as {\sc \#$k$-Vertex Cover}, {\sc \#$k$-Vertex Cover} and {\sc \#$k$-Vertex Cover}. 
\item {\sc \#$k$-Min $(s,t)$-Cut}: Given a graph $G$ and two distinct vertices $s,t\in V(G)$, output the number of minimum $(s,t)$-cuts in $G$. Here, the parameter $k$ is the size of a minimum $(s,t)$-cut in $G$.
\item {\sc \#$w$-Min $(s,t)$-Cut}: Defined as {\sc \#$k$-Min $(s,t)$-Cut} with the exception that the parameter $w$ is the treewidth of $G$.
\item {\sc \#$k$-Odd Cycle Transversal}: Given a graph $G$ and a non-negative integer $k$, output the number of odd cycle transversal of $G$ of size at most $k$. Here, the parameter is $k$.
\end{itemize}

\section{Kernelization of Counting Problems}\label{sec:kernelDef} 
We define the notion of kernelization for counting problems as follows.

\begin{definition}[{\bf Compression of Counting Problem}]
Let $P$ and $Q$ be two parameterized counting problems. A {\em compression} (or {\em compression algorithm}) of $P$ into $Q$ is a pair $(\mathsf{reduce},\mathsf{lift})$ of two polynomial-time procedures such that:
\begin{itemize}
\item Given an instance $(x,k)$ of $P$, $\mathsf{reduce}$ outputs an instance $(x',k')$ of $Q$ where $|x'|,k'\leq f(k)$ for some computable function $f$.
\item Given an instance $(x,k)$ of $P$, the instance $(x',k')$ that is the output of $\mathsf{reduce}$ on $(x,k)$, and $s'$ such that $P(x',k')=s'$, $\mathsf{lift}$ outputs $s$ such that $Q(x,k)=s$.
\end{itemize}
When $Q$ is immaterial, we refer to a compression of $P$ into $Q$ only as a compression of $P$.
\end{definition}

When $P=Q$, a compression is called a {\em kernel}. The measure $f(k)$ is termed the {\em size} of the compression. When $f$ is a polynomial function, then the compression (or kernel) is said to be a {\em polynomial compression} ({\em polynomial kernel}). The following observation is immediate.

\begin{observation}
Let $P$ be a parameterized (decision) problem that does not admit a polynomial kernel (or compression). Then, no counting version of $P$ admits a polynomial kernel (or compression).
\end{observation}

Hence, we only consider parameterized counting problems whose decisions versions are either in P, or, if they are not, then they at least admit polynomial kernels. Specifically, {\sc Min $(s,t)$-Cut} is in P~\cite{cormen2001introduction}, and polynomial kernels for {\sc $k$-Vertex Cover}, {\sc $\ell$-Vertex Cover} (and {\sc $m$-Vertex Cover}),  {\sc $k$-Planar ${\cal F}$-Deletion}, and {\sc $k$-Odd Cycle Transversal} can be found in \cite{buss1993nondeterminism}, \cite{kratsch2020representative}, \cite{fomin2012planar} and \cite{kratsch2020representative} respectively.

Throughout the paper, whenever we discuss a compression, we suppose (implicitly) that the compression is into a well-behaved problem, defined as follows:

\begin{definition}[{\bf Well-Behaved Problem}]
Let $P: \Sigma^\star\times\mathbb{N}_0\rightarrow \mathbb{N}_0$ be a parameterized counting problem. Then, $P$ is {\em well-behaved} if there exists a polynomial-time algorithm that, given $n\in\mathbb{N}$ in unary, outputs $N\in\mathbb{N}$ in binary with the following property: for every $(x,k)\in \Sigma^\star \times\mathbb{N}_0$ of size at most $n$, $P(x,k)\leq N$.
\end{definition}

We remark that, essentially, every ``natural'' parameterized counting problem (that we know of) is well-behaved.

\begin{restatable}{lemma}{countingFPTKernel}
\label{lem:countingFPTKernel}
 Let $P$ be a parameterized counting problem that is solvable in finite time. Then, $P$ is FPT if and only if it admits a kernel.
\end{restatable}
 
 \begin{proof}
 The proof follows lines similar to that of Proposition \ref{prop:FPTequivKer} (we also \cite{lokshtanov2017lossy}). For the sake of completeness, we present the details in Appendix \ref{sec:prelimsProofs}.
 \end{proof}

Due to Lemma \ref{lem:countingFPTKernel}, every counting problem that is \#W[1]-hard (and which is solvable in finite time) does not admit any kernel, even not of exponential (or worse) size. We remark that {\sc \#$k$-Min$(s,t)$-Cut} is shown to be FPT by Berge et al.~\cite{berge2019fixed}, and  {\sc \#$w$-Min$(s,t)$-Cut} is can be shown to be FPT by the usage of straightforward dynamic programming over tree decompositions (see, e.g. \cite{DBLP:books/sp/CyganFKLMPPS15}).

\begin{restatable}{lemma}{countingPPT}
\label{lem:countingPPT}
 Let $P, Q$ be two parameterized counting problems such that there exists a PPT from $P$ to $Q$. If $Q$ admits a polynomial compression, then $P$ admits a polynomial compression.
\end{restatable}
 
 \begin{proof}
 The proof follows lines similar to that of Proposition \ref{prop:PPT}. For the sake of completeness, we present the details in Appendix \ref{sec:prelimsProofs}.
 \end{proof}
 
We remark that in Sections \ref{sec:lowerSum} and \ref{sec:lower} we discuss two new notions of a cross-composition for proofs of the unlikely existence of polynomial compressions for parameterized counting problems.


\section{Polynomial Kernel for \#Vertex Cover}\label{sec:vc}

The purpose of this section is to prove the following theorem.

\vcKernel*

Towards the proof of this theorem, we first develop the reduction procedure. Then, we discuss properties of the reduced instance. Afterwards, we present the lifting procedure and conclude the correctness of the kernel. For the sake of brevity, throughout this section, we write {\sc \#Vertex Cover} instead of {\sc \#$k'$-Vertex Cover} (where $k'$ is the current value of the parameter).

\subsection{Reduction Procedure and a Corollary for Minimal Vertex Covers} 

We define the procedure $\mathsf{reduce}$ as follows. Given an instance $(G,k)$ of {\sc \#Vertex Cover}, we will first exhaustively apply the following reduction rule, known as Buss Rule~\cite{buss1993nondeterminism} (see also \cite{DBLP:books/sp/CyganFKLMPPS15}):

\begin{definition}[{\bf Buss Rule}]\label{BussRule}
If $G$ contains a vertex $v$ of degree at leas $k+1$, then update $G\leftarrow G-\{v\}$ and $k\leftarrow k-1$.
\end{definition}  

Let $(G_1,k_1)$ be the instance of {\sc \#Vertex Cover} obtained after exhaustive application of Buss Rule. Let $G_2$ be graph obtained from $G_1$ by the removal of all isolated vertices, and denote $k_2=k_1$. Let $\cal S$ (${\cal S}_1$, ${\cal S}_2$) denote the set of vertex covers of $G$ ($G_1$, $G_2$) of size at most $k$ ($k_1$, $k_2$) . Let $n_1=|V(G_1)|$ and $n_2=|V(G_2)|$. We have the following known proposition:

\begin{proposition}[\cite{buss1993nondeterminism,DBLP:books/sp/CyganFKLMPPS15}]\label{prop:BussRule}
The three following properties hold:
\begin{enumerate}
\item ${\cal S}=\{S_1\cup (V(G)\setminus V(G_1)): S_1\in {\cal S}_1\}$.
\item If $|E(G_2)|>(k_2)^2$, then $G$ does not contain any vertex cover of size at most $k$.
\item Else, $|E(G_2)|\leq (k_2)^2$, then $|V(G_2)|\leq 2(k_2)^2$.
\end{enumerate}
\end{proposition}

Let $x$ (resp., $x'$) be the number of vertex covers (resp., minimal vertex covers) of $G$ of size at most $k$, let $x_1$ (resp., $x_1'$) be the number of vertex covers (resp., minimal vertex covers) of $G_1$ of size at most $k_1$, and let $x_2$ (resp., $x_2'$) be the number of vertex covers (resp., minimal vertex covers) of $G_2$ of size at most $k_2$ .  Then, due to the first item of Proposition \ref{prop:BussRule} and since no minimal vertex cover can contain isolated vertices, we have the following corollary.

\begin{corollary}\label{cor:equalities}
The following equalities hold: $x=x_1$ and $x'=x'_1=x_2'$.
\end{corollary}

Given this corollary, we can already conclude a polynomial kernel for the variant of {\sc \#Vertex Cover} termed {\sc \#$k$-Minimal Vertex Cover}. (The challenge, dealt with in the rest of Section \ref{sec:vc}, would be to derive a polynomial kernel  for{\sc \#Vertex Cover} .)

\begin{theorem}
{\sc \#$k$-Minimal Vertex Cover} admits a kernel of size $\OO(k^2)$.
\end{theorem}

\begin{proof}
Given an instance $(G,k)$ of {\sc \#$k$-Minimal Vertex Cover}, the procedure $\mathsf{reduce}'$ outputs: {\em (i)} $(G_2,k_2)$ if $|E(G_2)|\leq (k_2)^2$, and {\em (ii)} $(G'=(\{u,v\},\{\{u,v\}\}), k'=0)$ otherwise. Observe that the procedure runs in polynomial time, and, due to the third item of Proposition \ref{prop:BussRule}, the size of the output is bounded by $\OO(k^2)$. 

Given $(G,k)$, the output of $\mathsf{reduce}'$, and the solution $x_2'$ to this output, the procedure $\mathsf{lift}'$ returns $x_2'$. Observe that the procedure runs in polynomial time, and from the second item of Proposition \ref{prop:BussRule} and Corollary \ref{cor:equalities}, we know that $x'=x_2'$ and hence the procedure is correct.
\end{proof}

Unfortunately, for {\sc \#Vertex Cover}, we cannot simply output $(G_2,k_2)$. In particular, observe that different vertex covers of $G_1$  of size at most $k_1$ might contain different numbers of vertices that are isolated in $G_1$, and hence the knowledge of $x_2$ alone is insufficient in order to deduce $x_1$ (and $x$).

We proceed to modify $G_2$ in order to define the graph that will be the output of the reduction

\begin{definition}\label{def:defineG3}
Let $d=n_2$ and $t=d+dk_2+2(dk_2)^2$. Then, let $G_3$ be the graph whose vertex set $\{v_i: v\in V(G_2), i\in [d]\}\cup T$, where $T$ is a set of $t$ new vertices,  and whose edge set is $\{\{u_i,v_j\}: \{u,v\}\in E(G_2), i,j\in[d]\}$. Additionally, let $k_3=d\cdot k_2$.
\end{definition}

That is, $G_3$ is the result of the replacement of every vertex of $G_2$ by $d$ copies (false twins) of that vertex and the addition of $t$ new vertices. We are now ready to define $\mathsf{reduce}$.

\begin{definition}[{\bf Procedure $\mathsf{reduce}$}]
Given an instance $(G,k)$ of {\sc \#Vertex Cover}, the procedure $\mathsf{reduce}$ outputs: {\em (i)} $(G_3,k_3)$ if $|E(G_2)|\leq (k_2)^2$, and {\em (ii)} $(G'=(\{u,v\},\{\{u,v\}\}), k'=0)$ otherwise.
\end{definition}

Due to the third item of Proposition \ref{prop:BussRule}, we have the following immediate observation.

\begin{observation}\label{obs:reduceVC}
$\mathsf{reduce}$ runs in polynomial time, and the size of its output is bounded by $k^{\OO(1)}$.
\end{observation}

\subsection{Properties of the Reduced Instance}

For every $i\in\{0,1,\ldots,k_2\}$, let ${\cal S}^i_2$ be the set of vertex covers of $G_2$ of size exactly $i$. Then, ${\cal S}_2=\bigcup_{i=0}^{k_2}{\cal S}^i_2$ is the set of vertex covers of size at most $k_2$ of $G_2$. Let ${\cal S}_3$ be the set of vertex covers of $G_3$ of size at most $k_3$. Let $n_3=|V(G_3)|$. We say that a subset $U\subseteq V(G_3)$ is {\em valid} if there does not exist $v\in V(G_2)$ such that $\{v_1,v_2\ldots,v_d\}\subseteq U$. We proceed to define the following mappings.

\begin{definition}[{\bf Mappings $\mathsf{map}$ and $\mathsf{Map}$}]\label{def:map}
The mappings $\mathsf{map}: {\cal S}_2\rightarrow 2^{V(G_3)}$ and $\mathsf{Map}{\cal S}_2\rightarrow 2^{2^{V(G_3)}}$ are defined as follows.
\begin{itemize}
\item Given $X\in{\cal S}_2$, let $\mathsf{map}(X)=\{v_j: v\in X, j\in[d]\}$.
\item Given $X\in{\cal S}_2$, let $\mathsf{Map}(X)=\{\mathsf{map}(X)\cup U: U$ is valid, $U\cap \mathsf{map}(X)=\emptyset, |U|\leq k_3-|\mathsf{map}(X)|\}$.
\end{itemize}
\end{definition}

We have the following lemma regarding the vertex covers of $G_3$.

\begin{lemma}\label{lem:union}
We have that {\em (i)} ${\cal S}_3=\bigcup_{X\in{\cal S}_2}\mathsf{Map}(X)$, and {\em (ii)} for distinct $X,Y\in{\cal S}_2$, $\mathsf{Map}(X)\cap \mathsf{Map}(Y)=\emptyset$.
\end{lemma}

\begin{proof}
We first prove the correctness of the first item. On the one hand, consider some $A\in{\cal S}_3$. Let $X=\{v\in V(G_2): \{v_1,v_2,\ldots,v_d\}\subseteq X\}$, and $U=A\setminus X$. We claim that $X\in{\cal S}_2$. Since $|A|\leq k_3$ (because $A\in{\cal S}_3$) and $k_3=d\cdot k_2$, it follows that $|X|\leq k_2$. Moreover, consider an edge $\{u,v\}\in E(G_2)$. If there exists $u_i,v_j\in V(G_3)$ such that $\{u_i,v_j\}\cap A=\emptyset$, then we have a contradiction since $A$ is a vertex cover of $G_3$. Hence, $\{u,v\}\cap X\neq\emptyset$. In turn, we derive that $X$ is a vertex cover of $G_2$, which yields that $X\in{\cal S}_2$. 
Now, notice that, by Definition \ref{def:map}, $A=\mathsf{map}(X)\cup U$ and $\mathsf{map}(X)\cup U\in\mathsf{Map}(X)$. So, $A\in \bigcup_{X\in{\cal S}_2}\mathsf{Map}(X)$.

On the other hand, let $B\in \bigcup_{X\in{\cal S}_2}\mathsf{Map}(X)$. So, $B\in \mathsf{Map}(X)$ for some $X\in {\cal S}_2$. By Definition \ref{def:map}, this implies that $B=\mathsf{map}(X)\cup U$ for some valid subset $U\subseteq V(G_3)$ disjoint from $\mathsf{map}(X)$, and $|B|\leq k_3$. So, to derive that $B\in{\cal S}_3$, it suffices to argue that $\mathsf{map}(X)$ is a vertex cover of $G_3$. To this end, consider some edge $\{u_i,v_j\}\in E(G_3)$. Then, $\{u,v\}\in E(G_2)$. Because $X$ is a vertex cover of $G_2$, we have that $\{u,v\}\cap X\neq\emptyset$. However, by Definition \ref{def:map}, this implies that $\{u_i,v_j\}\cap\mathsf{map}(X)\neq\emptyset$. Thus, the proof of the first item of the lemma is complete.

For the second item of the lemma, consider some distinct $X,Y\in{\cal S}_2$. Without loss of generality, suppose that $|X|\geq|Y|$. So, there exists $v\in V(G_2)$ such that $v\in X\setminus Y$, and, hence, $\{v_1,v_2,\ldots,v_d\}\subseteq \mathsf{map}(X)$ while $\{v_1,v_2,\ldots,v_d\}\cap \mathsf{map}(Y)=\emptyset$. So, since a valid set cannot contain $\{v_1,v_2,\ldots,v_d\}$, we derive that $\{v_1,v_2,\ldots,v_d\}\setminus A\neq\emptyset$ for every $A\in\mathsf{Map}(Y)$. 
However, since $\{v_1,v_2,\ldots,v_d\}\setminus A=\emptyset$ for every $A\in\mathsf{Map}(X)$, it follows that $\mathsf{Map}(X)\cap\mathsf{Map}(Y)=\emptyset$.
\end{proof}

We consider the sizes of the sets assigned by $\mathsf{Map}$ in the following lemma.

\begin{lemma}\label{lem:sizes}
For every $X\in{\cal S}^i_2$ for $i\in\{0,1,\ldots,k_2\}$, it holds that $|\mathsf{Map}(X)|=w_i$, where \[\displaystyle{w_i=\sum_{(a^\star,a_1,a_2,\ldots,a_{n_2-i})\in W_i}{t\choose a^\star}\prod_{j=1}^{n_2-i}{d \choose a_j}},\ \ \  \mathrm{and}\]
\[W_i=\{(a^\star,a_1,a_2,\ldots,a_{n_2-i}): \displaystyle{a^\star+\sum_{j=1}^{n_2-i}a_j}\leq k_3-d\cdot i, a^\star\leq t, \mathrm{and\ for\ each\ }j\in[n_2-i], a_j\in\{0,1,\ldots,d-1\}\}.\]
\end{lemma}

Towards the proof of this lemma and a latter lemma, for every $r\in\{0,1,\ldots,k_3-i\cdot d\}$, let us denote $W^r_i=\{(a^\star,a_1,a_2,\ldots,a_{n_2-i}): a^\star+\sum_{j=1}^{n_2-i}a_j=r, a^\star\leq t, \mathrm{and\ for\ each\ }j\in[n_2-i], a_j\in\{0,1,\ldots,d-1\}\}$, and $w^r_i=\sum_{(a^\star,a_1,a_2,\ldots,a_{n_2-i})\in W^r_i}{t \choose a^\star}\prod_{j=1}^{n_2-i}{d \choose a_j}$.

\begin{proof}[Proof of Lemma \ref{lem:sizes}]
Let $X\in{\cal S}^i_2$. So, we need to count the number of subsets $U\subseteq V(G_3)$ such that $U$ is valid, $U\cap\mathsf{map}(X)=\emptyset$, and $|U|\leq k_3-|\mathsf{map}(X)|$. Observe that $|\mathsf{map}(X)|=d\cdot i$. So, because we demand that $U\cap\mathsf{map}(X)=\emptyset$, every choice of $U$ corresponds to the choice of some $r\leq k_3-d\cdot i$ vertices from $V(G_3)\setminus\mathsf{map}(X)$ such that the resulting set would be valid. In turn, every such choice, for a specific $r$, corresponds to the choice of {\em (a)} how many vertices to pick from $T$ and how many vertices (a number between $0$ and $d-1$, due to validity) to pick from $\{v_1,v_2,\ldots,v_d\}$ for every $v\notin X$, so that in total we pick $r$ vertices, and {\em (b)} given a choice of type {\em (a)}, the choice of which specific vertices to pick from $T$ and which specific vertices to pick from $\{v_1,v_2,\ldots,v_d\}$ for every $v\notin X$. Clearly, we have a natural 1-to-1 correspondence between the the choices of type {\em (a)} and the vectors in $W^r_i$. Then, given a choice of such a vector $(a^\star,a_1,a_2,\ldots,a_{n_2-i})$, we have ${t\choose a^\star}\prod_{j=1}^{n_2-i}{d \choose a_j}$ choices of type {\em (b)}. Considering all choices for $r$, we attain the formula stated in the lemma.
\end{proof}

In particular, we prove that the sizes in Lemma \ref{lem:sizes} satisfy the following.

\begin{lemma}\label{lem:bound}
For every $i\in\{0,1,\ldots,k_2\}$,
\[w_i>\sum_{j=i+1}^{k_2}{n_2 \choose j}\cdot w_j.\]
\end{lemma}

\begin{proof}
Fix $i\in\{0,1,\ldots,k_2\}$. First, observe that $w_p\geq w_q$ for all $p,q\in\{0,1,\ldots,k_2\}$ such that $p\leq q$, and ${n_2\choose j}\leq 2^{n_2}$ for all $j\in\{0,1,\ldots,k_2\}$. Hence, it suffices to prove that $w_i\geq k_2\cdot 2^{n_2}\cdot w_{i+1}$. For this purpose, notice that $n_3=dn_2+t$.  Additionally, on the one hand, for all $i'\in\{0,1,\ldots,k_2\}$,
\[w_{i'}\leq k_3\cdot {n_3-d {i'} \choose k_3-d {i'}} = dk_2\cdot {d (n_2-{i'})+t \choose d(k_2-{i'})}.\]
We refer to this inequality as Inequality (1). To see its correctness, note that $w^r_{i'}$ is maximum when $r$ is maximum (restricted to $\{0,1,\ldots,k_3-di'\}$), i.e., when $r=k_3-d{i'}\leq k_3$. Hence, $w_{i'}\leq k_3\cdot w^{k_3-d{i'}}_{i'}$. Now, observe that $w^{k_3-d{i'}}_{i'}$ corresponds to the number of choices of $k_3-d{i'}$ elements out of a universe of size $n_3-d{i'}$ that satisfy particular restrictions. Specifically, we have a partition of the universe into $n_2-i'+1$ parts---one of size $t$ and the others of size $d$--- and we can pick at most $d-1$ elements from each of the parts of size $d$. In particular, this simply means that $w^{k_3-d{i'}}_{i'}$ is bounded from above by the number of choices of $k_3-d{i'}$ elements out of a universe of $n_3-d{i'}$ elements, which is ${n_3-d {i'} \choose k_3-d {i'}}$. Thus, Inequality (1) is correct.

On the other hand,
\[w_{i'}\geq {n_3-d {i'}-n_2 \choose k_3-d{i'}} = {d(n_2-{i'}) +t-n_2 \choose d(k_2-{i'})}.\]
We refer to this inequality as Inequality (2). To see its correctness, note that $w_{i'}\geq w^r_{i'}$ for all $r\in\{0,1,\ldots,k_3-di'\}$. So, in particular, $w_{i'}\geq w^{k_3-di'}_{i'}$. Recall the combinatorial interpretation of $w^{k_3-d{i'}}_{i'}$ discussed above for the correctness of Inequality (1). Now, out of that universe, suppose that we remove (arbitrarily) one element from each of the parts of size $d$---so, in total, we remove $n_2-i'$ elements. Then, we remove $i'$ additional elements. Hence, we remain with a universe of size $n_3-di'-n_2$. However, every choice of $k_3-di'$ elements from this universe satisfies the particular restrictions stated in the aforementioned combinatorial interpretation. Hence, $w^{k_3-d{i'}}_{i'}$ is bounded from below by the number of choices of $k_3-d{i'}$ elements out of a universe of $n_3-d{i'}-n_2$ elements, which is ${n_3-d {i'}-n_2 \choose k_3-d {i'}}$. Thus, Inequality (2) is correct.

Hence, having Inequality (2) and since $d=n_2$,
\[\begin{array}{lll}
\smallskip
w_i & \geq& \displaystyle{{d (n_2-i)+t -n_2 \choose d(k_2-i)}.}\\

\smallskip
& =& \displaystyle{\frac{ (d(n_2\!-\!i)+t-d(k_2\!-\!i))(d(n_2\!-\!i)+t-d(k_2\!-\!i)-1)\cdots(d(n_2\!-\!i)+t-d(k_2\!-\!i)-n_2+1)}
{(d(k_2-i))(d(k_2-i)-1)\cdots(d(k_2-i)-n_2+1)}}\\

&&\cdot \displaystyle{{d (n_2-i)+t-d \choose d(k_2-i)-d}}.
\end{array}\]
Recall that $t=d+dk_2+2(dk_2)^2$. 
So, for all $j\in[d]$, $d(n_2-i)+t-d(k_2-i)-j+1\geq 2(dk_2)^2\geq 2dk_2\cdot(d(k_2-i)-j+1)$. In particular,  we derive that
\[\begin{array}{l}
\displaystyle{\frac{ (d(n_2-i)+t-d(k_2-i))(d(n_2-i)+t-d(k_2-\!i)-1)\cdots(d(n_2-i)+t-d(k_2-i)-n_2+1)}
{(d(k_2-i))(d(k_2-i)-1)\cdots(d(k_2-i)-n_2+1)}}\\
\geq (2dk_2)^{n_2}> d(k_2)^2\cdot 2^{n_2}.
\end{array}\]

Hence, the calculation above implies that
\[\begin{array}{lll}
\smallskip
w_i & >& \displaystyle{d(k_2)^2\cdot 2^{n_2}\cdot {d (n_2-i)-d \choose d(k_2-i)-d}}\\
&\geq& k_2\cdot 2^{n_2}\cdot w_{i+1},
\end{array}\]
where the last inequality follows from Inequality (1). As discussed earlier, this completes the proof.
\end{proof}

\subsection{Procedure $\mathsf{lift}$ and Proof of Theorem \ref{thm:vcKernel}}\label{sec:vcLift}

We start with a computation of the values $w_i$, $i\in\{0,1,\ldots,k_2\}$, defined in Lemma \ref{lem:sizes}.

\begin{lemma}\label{lem:computewi}
There exists a polynomial-time algorithm that, given $i\in\{0,1,\ldots,k_2\}$ and having $t,d,k_2$ and $n_2$ at hand, outputs $w_i$. Here, the input numbers are encoded in unary, and the output number is encoded in binary.
\end{lemma}

\begin{proof}
Observe that $w_i=\sum_{r=0}^{k_3-id}w^r_i$. Hence, for the proof, it suffices to fix some $r\in\{0,1,\ldots,k_2-id\}$, and show how to compute $w^r_i$ in polynomial time. Now, denote $\ell=n_2-i$, $q=d-1$, and
\[\widehat{w}^{p}_i=\sum_{(a_1,a_2,\ldots,a_\ell)\atop\mathrm{s.t.}\  \sum_{j=1}^\ell a_j=p,\ \mathrm{and}\ \forall j\in[\ell],  a_j\in\{0,1,\ldots,q\}}\prod_{j=1}^\ell{d\choose a_j}.\]
Then, $w^r_i=\displaystyle{\sum_{a^\star=0}^t{t\choose a^\star} \widehat{w}^{r-a^\star}_i}$. So, for the proof, it suffices to fix some $a^\star\in\{0,1,\ldots,t\}$, and show how to compute $\widehat{w}^{p}_i$, for $p=r-a^\star$, in polynomial time.

In what follows, we employ dynamic programming to compute $\widehat{w}^{p}_i$. To this end, for every $\ell'\in[\ell]$ and $p'\in\{0,1,\ldots,\min(p,\ell'\cdot q)\}$, we allocate a table entry $\mathfrak{M}[\ell',p']$. 
We define (for the analysis):
\[W_{\ell',p'} = \sum_{(a_1,a_2,\ldots,a_{\ell'})\atop\mathrm{s.t.}\  \sum_{j=1}^{\ell'} a_j=p',\ \mathrm{and}\ \forall j\in[\ell'],  a_j\in\{0,1,\ldots,q\}}\prod_{j=1}^{\ell'}{d\choose a_j}.\]

The purpose of $\mathfrak{M}[\ell',p']$ would be to store $W_{\ell',p'}$. Then, since $\widehat{w}^{p}_i=W_{\ell,p}$, we would output $\mathfrak{M}[\ell,p]$.

The basis is when  $\ell'=1$. Then, for every $p'\in\{0,1,\ldots,\min(p,\ell\cdot q)\}$, we initialize $\mathfrak{M}[\ell',p']={d \choose p'}$.

Now, for every $\ell'\in[\ell]$ in increasing order, and every $p'\in\{0,1,\ldots,p\}$ in arbitrary order, we perform the following computation:
\[\mathfrak{M}[\ell',p']\leftarrow \sum_{s=0}^{\min(p',q)} {d \choose s} \cdot \mathfrak{M}[\ell'-1,p'-s].\]

Clearly, the computation can be performed in polynomial time (since the input numbers are encoded in unary, and the numbers stored in the table are encoded in binary).

Combinatorially, the interpretation of $W_{\ell',p'}$ is of the number of choices to pick exactly $p'$ elements from a universe that is partitioned into $\ell'$ parts of size $d$ each, such that we can pick at most $q$ elements from each part. Equivalently, we can consider the number of choices to pick exactly $s\leq p'$ elements from the last part of the universe, and then, for each such choice, we can consider the number of choices to pick exactly $p'-s$ additional elements from the remainder of the universe, such that we can pick at most $q$ elements from each part. This yields the following equality:
\[W_{\ell',p'}=\sum_{s=0}^{\min(p',q)} {d \choose s} \cdot W_{\ell'-1,p'-s}.\]
In turn, using straightforward induction, this equality yields the correctness of the computation.
\end{proof}

Now, we define $\mathsf{lift}$ as follows.

\begin{definition}[{\bf Procedure $\mathsf{lift}$}]
Given an instance $(G,k)$ of {\sc \#Vertex Cover}, the output of $\mathsf{reduce}$, and the solution $x^\star$ to this output, the procedure $\mathsf{lift}$  performs the following operations:
\begin{enumerate}
\item Initialize $\widehat{x}\leftarrow x^\star$.
\item For $i=0,1,\ldots,k_2$:
	\begin{enumerate}
	\item Use the algorithm in Lemma \ref{lem:computewi} to compute $w_i$.
	\item Let $y_i\leftarrow \lfloor\widehat{x}/w_i\rfloor$.
	\item Update $\widehat{x}\leftarrow\widehat{x}-y_i\cdot w_i$.
	\item Let $z_i\leftarrow y_i\cdot\sum_{j=0}^{k_2-i}{n_1-n_2\choose j}$.
	\end{enumerate}
\item Return $z=\sum_{i=0}^{k_2}z_i$.
\end{enumerate}
\end{definition}

We start the analysis with the following observation, whose correctness is immediate from Lemma~\ref{lem:computewi} and the definition of $\mathsf{lift}$.

\begin{observation}\label{obs:lift}
$\mathsf{lift}$ runs in polynomial time.
\end{observation}

For every $X\in {\cal S}_2$, define $\mathsf{Pull}(X)=\{X\cup U: U\subseteq V(G_1)\setminus V(G_2), |X|+|U|\leq k_1\}$. For the correctness of $\mathsf{lift}$, we prove the two following lemmas.

\begin{lemma}\label{lem:union2}
We have that {\em (i)} ${\cal S}_1=\bigcup_{X\in{\cal S}_2}\mathsf{Pull}(X)$, and {\em (ii)} for distinct $X,Y\in{\cal S}_2$, $\mathsf{Pull}(X)\cap \mathsf{Pull}(Y)=\emptyset$.
\end{lemma}

\begin{proof}
Recall that $G_2$ is obtained from $G_1$ be the removal of all isolated vertices, and that $k_2=k_1$. Hence, every vertex cover of $G_1$ of size at most $k_1$ is the union of two sets, $A$ and $B$, where $A$ is a vertex cover of $G_2$ of size at most $k_2$, and $B\subseteq V(G_1)\setminus V(G_2)$ is of size at most $k_1-|A|$. So, the first item follows, and the second item is immediate.
\end{proof}

\begin{lemma}\label{lem:yz}
For every $i\in\{0,1,\ldots,k_2\}$, we have that {\em (i)} $y_i=|{\cal S}^i_2|$, and {\em (ii)} $z_i=|\bigcup_{X\in{\cal S}^i_2}\mathsf{Pull}(X)|$. 
\end{lemma}

\begin{proof}
From Lemma \ref{lem:union}, we have that $x^\star=\sum_{X\in{\cal S}_2}|\mathsf{Map}(X)|=\sum_{i=0}^{k_2}\sum_{X\in{\cal S}^i_2}|\mathsf{Map}(X)|$. So, by Lemma \ref{lem:sizes}, we derive that $x^\star=\sum_{i=0}^{k_2}|{\cal S}^i_2|\cdot w_i$. Observe that for every $i\in\{0,1,\ldots,k_2\},$ $|{\cal S}^i_2|\leq {n_2\choose i}$. Hence, due to Lemma \ref{lem:bound}, it follows that for every $i\in\{0,1,\ldots,k_2\}$, $w_i\geq\displaystyle{\sum_{j=i+1}^{k_2}|{\cal S}^j_2|\cdot w_j}$. Given the manner in which $\mathsf{lift}$ handles the variables $\widehat{x}$ and $y_0,y_1,\ldots,y_{k_2}$, this implies the correctness of the first item of the lemma.

Now, observe that for any $X\in{\cal S}^i_2$, $|\mathsf{Pull}(X)|=\sum_{j=0}^{k_2-i}{n_1-n_2 \choose j}$, and from the second item of Lemma \ref{lem:union2}, it follows that $|\bigcup_{X\in{\cal S}^i_2}\mathsf{Pull}(X)|=\sum_{X\in{\cal S}^i_2}|\mathsf{Pull}(X)|$. From these arguments, and since we have already proved the correctness of the first item, we derive the correctness of the second item as well.
\end{proof}

Having Corollary \ref{cor:equalities} and Lemmas \ref{lem:union2} and \ref{lem:yz} at hand, we prove the following lemma, which implies the correctness of $\mathsf{lift}$.

\begin{lemma}\label{lem:lift}
We have that $|{\cal S}|=z$.
\end{lemma}

\begin{proof}
By Corollary \ref{cor:equalities}, $|{\cal S}|=|{\cal S}_1|$. From Lemma \ref{lem:union2}, we have that $|{\cal S}_1|=|\bigcup_{X\in{\cal S}_2}\mathsf{Pull}(X)|$, which equals $\sum_{i=0}^{k_2}|\bigcup_{X\in{\cal S}^i_2}\mathsf{Pull}(X)|$. Further, from Lemma \ref{lem:yz}, we have that $\sum_{i=0}^{k_2}|\bigcup_{X\in{\cal S}^i_2}\mathsf{Pull}(X)|=\sum_{i=0}^{k_2}z_i=z$. So, we conclude that $|{\cal S}|=z$.
\end{proof}

Thus, the correctness of Theorem \ref{thm:vcKernel} follows from Observations \ref{obs:reduceVC} and \ref{obs:lift}, and Lemma \ref{lem:lift}.


\newcommand{\folio}{{\mathsf{folio}}}
\newcommand{\ccount}{{\mathsf{count}}}

\section{Polynomial Compression for \#Planar ${\cal F}$-Deletion}\label{sec:compression}

In this section we present a polynomial compression for the {\sc \#Planar-${\cal F}$-Deletion} problem, which is a general problem encompassing {\sc \#Vertex Cover}, {\sc \#Feedback Vertex Set} and many others~\cite{fomin2012planar}.
Let us begin by recalling the {\sc Planar-${\cal F}$-Deletion} problem, where $\cal F$ is a finite set of connected graphs with at least one planar graph. The input is a graph $G$ and an integer $k$. The objective is to determine if there is a subset $S$ of at most $k$ vertices such that $G-S$ is $\cal F$-minor free.
In the counting version of the problem, {\sc \#Planar-${\cal F}$-Deletion}, given $G$ and $k$ we must output the number of distinct vertex subsets $S$ such that $|S| \leq k$ and $G-S$ is $\cal F$-minor free. We prove the following theorem in this section.

\compression*

At a high level, we follow the approach of \cite{fomin2012planar} which gave a polynomial kernel for {\sc \#Planar-${\cal F}$-Deletion}, but we develop additional results that allow us to compress and then recover the number of solutions of size $k$. We note that we only obtain a compression, and not a kernel, unlike the results for {\sc \#Vertex Cover} presented earlier.

\subsection{Preliminaries}\label{sec:compression-prelims}

We say that $S \subseteq V(G)$ is a $\cal F$-deletion set of $G$, if $G-S$ is $\cal F$-minor free. We enumerate a few properties of $\cal F$-minor free graphs.

\begin{proposition}[\cite{fomin2012planar}~Proposition~1]\label{prop:PFD-tw} If a graph $G$ is $\cal F$-minor free, where ${\cal F}$ is a finite family of graphs containing at least one planar graph, then there is a constant $h$ depending only on $\cal F$ such that $tw(G) \leq \eta$.
\end{proposition}

Let $(G,k)$ denote the input instance of {\sc \#Planar-$\cal F$ Deletion}. Following~\cite{fomin2012planar}, the first step of our {\sf reduce} algorithm is to compute a modulator to $(G,k)$ using an approximation algorithm for {\sc Planar-$\cal F$-Deletion}.

\begin{proposition}[\cite{fomin2012planar}]\label{prop:PDF-approx}There is a randomized polynomial time algorithm that given an instance of $(G,k)$ of {\sc Planar-$\cal F$ Deletion} either outputs a solution of size at most $c \cdot k$ for a fixed constant $c$ that depends only on $\cal F$, or correctly reports that no solution of size $k$ exists for $(G,k)$. This algorithm succeeds with probability at least $1 - 1/2^n$.
\end{proposition}

Having computed the approximate solution $X$, we first check if $|X| \leq c(k+1)$. If not, then it follows that the $(G,k)$ admits no solutions of size $k$.
Otherwise, $(G,k)$ admits a modulator of size at most $c(k+1)$, which we denote by $X$. Note that the bound was chosen as $c(k+1)$ instead of $ck$ to be consistent with \cite{fomin2012planar}. Observe that the graph $G-X$ is $\cal F$-minor free. Recall that, by Proposition~\ref{prop:PFD-tw}, the treewidth of any $\cal F$-minor free graph is upper-bounded by a constant $\eta$ that depends only on $\cal F$. We augment $X$ with additional vertices to arrive at the following.

\begin{proposition}[\cite{fomin2012planar}~Lemma~25,~26]~\label{prop:pfd-modulator}
    There is a randomized polynomial time algorithm that given an instance $(G,k)$ of {\sc Planar $\cal F$-Deletion}, either returns that $(G,k)$ has no solutions of size $k$, or computes two disjoint vertex subsets $X$ and $Z$, with probability at least $1-1/2^n$ such that,
    \begin{itemize}
        \item $|X| = \OO(k)$ and $|Z| = \OO(k^3)$,
        \item $X$ is a $\cal F$-deletion set of $G$
        \item For every connected component $C$ of $G - (X \cup Z)$, $|N(C) \cap Z| \leq 2(\eta + 1)$
        \item For any two vertices $u,v \in N(C) \cap X$, there are at least $k + \eta + 3$ vertex disjoint paths from $u$ to $v$ in $G - X$.
        \item For any $\cal F$-deletion set $S$ of size $k$, $|(N(C) \cap X) \setminus S| \leq \eta+1$.
    \end{itemize}
\end{proposition}

We call $X \cup Z$ an \emph{enriched modulator} to $(G,k)$.
Our next step is to compress the graph $G - (X \cup Z)$.
This is accomplished in two steps. First we reduce the number of connected components in $G - (X \cup Z)$ to $k^{\OO(1)}$ and then we store each component in a compressed form that is sufficient to count the number of $k$-size solutions of $(G,k)$. Let us introduce some additional notation from \cite{fomin2012planar,fomin2019kernelization} that are required for these results.


A {\bf boundaried graph} is a graph $G$ with a set of distinguished vertices $B$ and an injective mapping $\lambda_G$ from $B$ to $\mathbb{Z}^+$. The set $B$ is called the {\bf boundary} of $G$ which is also denoted by $\delta(G)$, and $\lambda_G$ is called the {\bf labelling} of $G$. The {\bf label-set} of $G$ is $\Lambda(G) = \{ \lambda(v) \forall v \in \delta(G)\}$. Given a finite set $I \subseteq \mathbb{Z}^+$, let ${\cal G}_I$ denote the set of all boundaried graphs whose label-set is $I$. Let ${\cal G}_{\subseteq I}$ denotes all boundaried graphs whose label-set is a subset of $I$. Finally, for $t \in \mathbb{Z}^+$, $G$ is a $t$-boundaried graph is $\Lambda(G) \subseteq \{1,2, \ldots, t\}$.

The {\bf gluing operation $\oplus$} on two $t$-boundaried graphs $G$ and $H$ gives the (non boundaried) graph $G \oplus H$ obtained by taking the disjoint union of $G$ and $H$ and then identifying pairs of vertices in $\delta(G)$ and $\delta(H)$ with the same label, and finally forgetting all the labels.
The {\bf boundaried gluing operation $\oplus_\delta$} is similar, but results in a boundaried graph: given two $t$-boundaried graphs $G$ and $H$, the $t$-boundaried graph $G \oplus_\delta H$ is obtained by taking the disjoint union of $G$ and $H$ and then identifying pairs of vertices in $\delta(G)$ and $\delta(H)$ with the same label; this results in $t$ new vertices that form the boundary of the new graph.

A $t$-boundaried graph $H$ is a \emph{minor} of a $t$-boundaried graph $G$, if $H$ is a minor of $G$ that is obtained without contracting any edge whose both endpoints are boundary vertices. Note that, if we contract an edge with exactly boundary vertex as an endpoint, the new vertex is also a boundary vertex with the same label. This relation is denoted by $H \leq_m G$.
The {\bf folio} of a $t$-boundaried graph $G$ is $\folio(G) = \{H \leq_m G\}$.
For two vertex subsets $P, B \subseteq V(G)$, $G^B_P$ denoted the $|B|$-boundaried graph $G[B \cup P]$ with $B$ as the boundary.

For a parameterized graph problem $\Pi$, we define an {\bf equivalence relation $\equiv_{\Pi}$} on the class of $t$-boundaried graphs as follows.
Two $t$-boundaried graphs $G_1$ and $G_2$ are equivalent if and only if the following holds: for any other $t$-boundaried graph $G_3$, $(G_1\oplus G_3, k) \in \Pi$ if and only if $(G_2\oplus G_3, k+c) \in \Pi$, where $c$ is a constant for $\Pi$.
We say that $\Pi$ has {\bf Finite Integer Index}, the equivalence relation $\equiv_{\Pi}$ partitions ${\cal G}_t$ into finitely many equivalence classes.
We shall require stronger conditions on the constant $c$ for kernelization. Towards this, we say that $\Pi' \subseteq \sigma^* \times \mathbb{Z}$ is a {\bf (positive) extended parameterized problem} of $\Pi$, if $(I,k) \in \Pi'$ whenever $k \leq 0$ and $\Pi' \cap (\Sigma^* \times \mathbb{Z}^+) = \Pi$.
Note that the extended parameterized problem $\Pi'$ of $\Pi$ is unique.
Next, consider an equivalence class $\cal R$ of $\equiv_{\Pi'}$ that is a subset of ${\cal G}_t$. We say that $H \in {\cal R}$ is a {\bf progressive representative} of $\cal R$ if for any $G \in {\cal R}$ and any $t$-boundaried graph $G'$, $(G \oplus G',k) \in \Pi'$ if and only if $(H \oplus G', k+c)$ in $\Pi$ such that $c \leq 0$.
We have the following proposition, that ensures the existence of progressive representatives for those $\Pi$ that admit an extension.

\begin{proposition}[\cite{fomin2019kernelization} Lemma 16.11]\label{prop:pfd-pro-rep}
    Let $\Pi$ be an extended parameterized graph problem. Then each equivalence class of $\equiv_{\Pi}$ has a progressive representative.
\end{proposition}

From now onwards, let us fix $\Pi$ to be {\sc Planar $\cal F$-deletion}, and let $\equiv_{\cal F}$ denote the equivalence relation for this problem.
We have the following proposition.
\begin{proposition}[\cite{fomin2012planar} Proposition 2]\label{prop:pdf-fii}
    If $\cal F$ is a finite family of connected graphs then {\sc $\cal F$-Deletion} has finite integer index.
\end{proposition}

Let ${\cal S}_t$ denote the set that contains one progressive representative for each equivalence class of $\equiv_{\cal F}$ that is a subset of ${\cal  G}_t$. Observe that, for each $t \in \mathbb{Z}^+$ the set ${\cal S}_t$ has constant cardinality that is equal to equal to the number of equivalence classes of $\equiv_{\cal F}$ in $G_{t}$. We define ${\cal S}_{\leq t} = \cup_{t' \leq t}~{\cal S}_{t'}$. Let $c_{t,\cal F} = |{\cal S}_{\leq t}$, which is a constant that depends only on $\cal F$ and $t$. Furthermore, the sizes of the the graphs in ${\cal S}_t{\leq t}$ is also a constant that depends only on $\cal F$ and $t$.

Let $h$ be the maximum number of vertices in a graph in $\cal F$.
For a component $C$ of $G - (X \cup Z)$, the {\bf border collection} ${\cal B}_C$ of $C$ is the collection of all vertex subsets $B$ such that ~$(i)$~$B \setminus X \subseteq N(C) \setminus X$ and ~$(ii)$~$|B \cap X| \leq \eta + 1$.
For a set $B \in {\cal B}_{C}$ and a boundaried graph $H$ with $B$ as the boundary, we say $C$ {\bf realizes} $(B,H)$ if $H \leq_m G^B_C$. Observe that, in this case $B \subseteq X \cup Z$, and $|X \cap B| \leq \eta + 1$ and hence $|B \setminus X| \leq |N(C) \cap Z| \leq 2(\eta + 1)$ using the bound from Proposition~\ref{prop:pfd-modulator}.
Let ${\cal B} = \bigcup_{\text{ component } C}{\cal B}_C$, and note that $|{\cal B}| \leq (|X| + |Z|)^{3(\eta + 1)}$ which is an upper-bound on the total number of possible borders over all components of $G - (X \cup Z)$.
For our purposes it is sufficient to consider all graphs $H$ that contain at most $h + 3(\eta + 1)$ vertices; in particular this includes all possible
subgraphs of the graphs in $\cal F$. The number of such graphs is at most $2^{h+3(\eta + 1) \choose 2}$, which is a constant depending only on $\cal F$.


\subsection{The reduce Procedure}

Let us now turn to the {\sf reduce} procedure for {\sc \#Planar $\cal F$-Deletion}. As in \cite{fomin2012planar}, we start with an enriched modulator $(X \cup Z)$ for the instance $(G,k)$ given by Proposition~\ref{prop:pfd-modulator}. We then compress the remaining graph $G - (X \cup Z)$ in two parts. First, we bound the number of connected components by identifying and deleting certain irrelevant components that will always have an empty intersection with a minimal $\cal F$-deletion set of $G$ of size at most $k$. Let $(G',k)$ denote the resulting instance. The second step is to store a compressed representation of each connected component of $G - (X \cup Z)$ that is sufficient to count the number of solutions of size $k'$ for each $k' \leq k$ for $(G',k)$. The {\sf lift} procedure will then use this information to count the number of solutions of $(G,k)$ in polynomial time; we present it in the next section. Note that we assume that Proposition~\ref{prop:pfd-modulator} gives $X \cup Z$ of cardinality $\OO(k^3)$ in the rest of this section.

\subsubsection{Bounding the number of connected components}\label{sec:compression-phase-1}

If we have a large number of components in $G - (X \cup Z)$ then the following lemma allows us to identify an irrelevant one, that contributes no vertices to a minimal $\cal F$-deletion set of size at most $k$.
Consider a pair $(B,H)$ where $B \in {\cal B}$ and $H$ is a boundaried graph on at most $h + 3(\eta + 1)$ vertices with $B$ as it's boundary.
We say that a pair $(B,H)$ is \emph{rich} if there are at least $\tau_{rich} = |X| + |Z| + k + (h + 3(\eta + 1))^2 + 2$ components of $G - (X\cup Z)$ realizing it. The following lemma allows us to identify certain components of $G - (X \cup Z)$ as irrelevant. Note that $\tau_{rich} = \OO(k^3)$.

\begin{lemma}[\cite{fomin2012planar} Lemma 36]\label{lemma:pdf-irr-comp}
    Let $C$ be a component of $G - (X \cup Z)$ such that every pair $(B,H)$ that $C$ realizes is rich. Then $G$ has an $\cal F$-deletion set of size $k$ if and only if $G - V(C)$ does.
\end{lemma}

It is immediate from Lemma~\ref{lemma:pdf-irr-comp} that if $S$ is a minimal $\cal F$-deletion set of $G$ of size at most $k$ and $C$ is a rich component, then $S \cap V(C) = \emptyset$. Recall that, the number of choices for $B$ is at most $(|X| + |Z|)^{3(\eta + 1)}$, while the number of choices of graph $H$ for each $B$ is at most $2^{h+3(\eta + 1) \choose 2}$. For each pair $(B,H)$ and a component $C$, we can encode in MSOL if $H \leq_m G^B_C$, and test if it is true in linear time~\cite{fomin2012planar}. Hence in polynomial time we can test if every pair realized by a component $C$ is rich. We then arrive at the following reduction rule from \cite{fomin2012planar}.
\begin{reduction}\label{rr:pfd-irr}
    If every pair $(B,H)$ realized by a component $C$ of $G-(X \cup Z)$ is rich, then delete $V(C)$ from $G$.
\end{reduction}

The correctness of the above reduction rule is immediate from Lemma~\ref{lemma:pdf-irr-comp}. When the above reduction rule is not applicable, the total number of components in $G - (X \cup Z)$ is bounded by $\tau_{\#comp} = \tau_{rich} \cdot (|X| + |Z|)^{3(\eta + 1)} \cdot 2^{h+3(\eta + 1) \choose 2}$~(\cite{fomin2012planar}~Lemma~36). Note that $\tau_{\#comp} = \OO(k^{3+9(\eta + 1)})$, as $h$ and $\eta$ are constants depending only on $\cal F$.

\subsubsection{Compressing the connected components}\label{sec:compression-phase-2}
In this section we show how we can store a compressed representation of each connected component $C$ of $G - (X \cup Z)$ that is sufficient to count the number of $\cal F$-deletion sets of $G$ of size at most $k'$ for any $k' \leq k$. Throughout this section, we assume that $k^{\tau_n} \geq \log n$ where $\tau_n$ is a constant depending on $\cal F$ that will be specified later.
We will justify this assumption in the description of {\sf lift} procedure, where we will argue that if $n$ is too large then the we can count all $\cal F$-deletion sets of size $k$ in polynomial time.

Consider a component $C$ of $G - (X \cup Z)$, and some $\cal F$-deletion set $S$ of size at most $k$ in $G$. It follows from Proposition~\ref{prop:pfd-modulator} that $tw(C) \leq \eta$, and $|N(C) \cap ((X \cup Z) \setminus S)| \leq 3(\eta + 1)$.
In essence, $C$ is a \emph{near-protrusion} of $G$ as defined in~\cite{fomin2012planar}.
For normal kernelization it is sufficient to identify an an irrelevant vertex or edge in this component if it were too large. For counting kernelization (compression), we must store information about all possible ways that $S$ and $C$ intersect. Therefore we need more detailed information about the what the ``structure'' of $C - S$ could be, and how many ways is it possible to attain this structure by deleting vertices in $N[C]$.

More precisely, consider a subset $S$ of size at most $k$, in $G - S$, consider the boundaried graph $G[N[C] \setminus S]$ with boundary $N(C) \setminus S$. Note the following associated properties:
\begin{itemize}
    \item The number $i_C = |C \cap S|$, which is a one of $\{0,1,\ldots, k\}$. This denotes the number of vertices from $C$ that is picked into $S$.

    \item The boundary $B_{S,C}$ of $G[N[C] \setminus S]$, i.e.  $B_{S,C} = N(C) \setminus S$, which is a subset of $X \cup Z$ of size at most $3(\eta + 1)$. Recall that $|X \cup Z| \leq \OO(k^3)$ and hence the number of possibilities for $N(C) \setminus S$ is at most $k^{9(\eta+1)}$.

    \item Finally, the equivalence class ${\cal R}$ of $\equiv_{\cal F}$ that contains the boundaried graph $G[N[C] \setminus S]$ with boundary $N(C) \setminus S$. Observe that the size of the boundary is at most $t = 3(\eta + 1)$, and there are at most $c_{t, \cal F}$ choices of $\cal R$, which is a constant dependent only on $\cal F$.
\end{itemize}
Let $S$ be a subset of at most $k$ vertices.
We say that the {\bf signature of $S$ with respect to $C$}, denoted $\sigma(S,C)$, is the tuple $(i_C, B_{S,C}, {\cal R})$ of the terms defined above. The {\bf signature of $S$}, $\sigma(S)$, is the collection of $\{\sigma(S,C) ~~\forall \text{ component } C\}$ along with $S \cap (X \cup Z)$.
For each component $C$, we store a table $T_C$ that for each possible choice of the tuple $\sigma(S,C)$, stores the number of subsets $S_C \subseteq V(C)$ that satisfy: $|S_C| = i_C$ and the graph $G[(V(C) \setminus S_C) \cup B_{S,C}]$ lies in the equivalence class $\cal R$.  Note that the  $T_C$ has at most $\tau_{table} = k \cdot k^{9(\eta+1)} \cdot c_{t, \cal F}$ entries, which is upper-bounded by a polynomial function of $k$. Further, in each entry of $T_C$ we store a number of value at most $n^k$. As $\log n \leq k^{\tau_n}$, we need at most $k^{\tau_n + 1}$ bits to store this number. Overall, we can store each table in $k^{\tau}+1 \cdot \tau_{table}$ space.
To compute the table $T_C$ for a component $C$, we have the following lemma, which intuitively applies a variant of Courcelle's theorem~\cite{DBLP:books/sp/CyganFKLMPPS15} with the treewidth as the parameter. Since the treewidth of $C$ is a constant ($\eta$) it runs in polynomial time.

\begin{lemma}\label{lemma:table}
    The table $T_C$ corresponding to the component $C$ can be computed in polynomial time.
\end{lemma}
\begin{proof}
    To compute the table $T_C$, for each tuple $(i_C, B_{S,C}, R)$ we need to compute and store the number of subsets $S_C \subseteq V(C)$ such that $|S_C| \leq i_C$ and $G[(V(C) \setminus S_C) \cup B_{S,C}]$ lies in the equivalence class $\cal R$ of $\equiv_{\cal F}$. The second condition can be expressed as an CMSO-formula $\psi$\cite{fomin2012planar}. Then using a dynamic programming algorithm we can count the number of $S_C \subseteq V(C)$ that satisfies $\psi$ and $|S_C| \leq i_C$. This dynamic programming algorithm implements an optimization version of Courcelle's Theorem~\cite{fomin2012planar}, except that it counts the number of solutions of size at most $i_C$. This runs in time exponential in $tw(G[V(C) \cup B_{S,C}]) \leq 4(\eta + 1) + |\psi|$, but polynomial in $|V(C)| + |B_{S,C}|$. Hence, for a fixed family $\cal F$, the algorithm runs in polynomial time.
\end{proof}

Let us next argue that the collection of tables $\{ T_C\}$ is sufficient to count the number of solutions of size at most $k'$ for each $k' \leq k$. Let us start with the following observation.


\begin{observation}\label{obs:simple-surgery}
    Let $S$ be a $\cal F$-deletion set of size at most $k$ in $G$. Let $C$ be a component of $G - (X \cup Z)$ and $S_C = S \cap C$. Let $i_C = |S_C|$, $B_{S,C} = N(C) \setminus X$ and $G_{S,C}$ be the boundaried graph $G[N[C] \setminus S]$ with boundary $B_{S,C}$. Let $G_{S,C}$ lie in the equivalence class ${\cal R}_C$ of $\equiv_{\cal F}$, and let $H_{C}$ be the progressive representative of ${\cal R}_C$. Let $\widehat{G_{S,C}}$ denote the boundaried graph $(G - (S \cup V(C))$ with boundary $B_{S,C}$. Then $\widehat{G_{S,C}} \oplus H_{C}$ is also $\cal F$-minor free.
\end{observation}
\begin{proof}
    Since $G-S = G_{S,C} \oplus \widehat{G_{C}}$ is $\cal F$-free, we have $(G_{S,C} \oplus \widehat{G_{S,C}}, 0) \in \Pi$ where $\Pi$ denotes the parameterized graph problem {\sc $\cal F$-Deletion}. Then, as $H_{C}$ is a progressive representative $(H_{C} \oplus \widehat{G_{S,C}}, 0) \in \Pi$.
\end{proof}

Next we attempt to characterize $\cal F$-deletion sets of size $k$ using progressing representatives. Towards this, let $S$ be a subset of at most $k$ vertices of $G$. For each component $C$ of $G - (X \cup Z)$, let $S_C = S \cap C$, $B_{S,C} = N(C) \setminus X$ and $G_{S,C}$ be the boundaried graph $G[N[C] \setminus S]$ with boundary $B_{S,C}$.
Let $G_{S,C}$ lie in the equivalence class ${\cal R}_C$ of $\equiv_{\cal F}$, and let $H_C$ be a progressive representative of ${\cal R}_C$.
We then define the following graph,
$$\bigoplus_{\text{component }C} G[(X \cup Z)\setminus S] \oplus_\delta H_C$$
where the gluing operations treats the graphs as $|X \cup Z|$-boundaried, with each vertex in $(X \cup Z)$ being labeled consistently across these graphs.
This requires that we first fix a labeling $\lambda_{XZ}$ of vertices in $X \cup Z$, and then for each $H_C$, we update it's labeling function $\lambda_{H_C}$ to be a restriction of $\lambda_{XZ}$ to $\delta(H_C)$, using the labels of $G_{S,C}$ as a guide. The overall effect is that each $G_{S,C}$ is replaced with it's progressive representative $H_C$ for all components $C$ of $G - (X \cup Z)$.

\begin{observation}\label{obs:full-surgery}
    Let $S$ be a subset of at most $k$ vertices of $G$.
    For each component $C$ of $G - (X \cup Z)$, let $S_C = S \cap C$, $B_{S,C} = N(C) \setminus X$ and $G_{S,C}$ be the boundaried graph $G[N[C] \setminus S]$ with boundary $B_{S,C}$.
    Let $G_{S,C}$ lie in the equivalence class ${\cal R}_C$ of $\equiv_{\cal F}$, and let $H_C$ be a progressive representative of ${\cal R}_C$.
    Then, $S$ is a $\cal F$-deletion set of $G$ if and only if the following graph is $\cal F$-minor free.
    $$\bigoplus_{\text{component }C} G[(X \cup Z)\setminus S] \oplus H_C$$.
\end{observation}
\begin{proof}
    This observation follows easily by iteratively applying Observation~\ref{obs:simple-surgery}, until every $G_{S,C}$ has been replaced with $H_C$.
\end{proof}


We are now ready to show that the tables $\{T_C\}$ are sufficient to count the number of $\cal F$-deletion sets of size $k'$ in $G$ for any $k' \leq k$.
Consider the following algebraic expression; we will prove that it computes the number of solutions of size at most $k'$ for the graph $G$.
For a logical statement $\phi$, Let $[\phi]$ be the $0,1$-indicator function which is $1$ if and only if $\phi$ is true.
\begin{align*}
    \ccount(k') = & \sum_{S_U \subseteq X \cup Z, ~|S_U| \leq k'} ~\sum_{\{\text{eq-class }{\cal R}_C ~\forall \text{ component }C\}} ~\sum_{\{i_C ~\forall \text{component }C ~\mid~ |S_U| + \sum_C i_C \leq k'\}} \\ &
    ~~~~~~~~\left[\bigoplus_{\text{component }C} G[(X \cup Z)\setminus S] \oplus H_C \text{ ~~~is $\cal F$-free}\right] \cdot \\
    & \prod_{\text{component } C} \left[|N(C) \setminus S_U| \leq 3(\eta + 1)\right]  \cdot T_C[{\cal R}_C, i_C, N[C] \setminus S_U]
\end{align*}

\begin{lemma}\label{lemma:count}
    $\ccount(k')$ counts the number of $\cal F$-deletion sets of $G$ of size at most $k'$.
\end{lemma}
\begin{proof}
    Consider the expression $\ccount(k')$ in the fully expanded as a summation over all choices of $S_U$, $\{{\cal R}_C ~\forall C\}$ and $\{i_C ~\forall C\}$. We need to verify that only the $\cal F$-deletion sets of size at most $k'$ contribute to this summation and each such set contributes $1$. Towards this, for subset $S$, let $S_U = S \cap (X \cup Z)$. For each component $C$ we have the equivalence class ${\cal R}_C$ of the boundaried graph $G[N[C] \setminus S]$ with boundary $N(C) \setminus S$ and $i_C = |S \cap C|$, and let $H_C$ be the progressive representative of ${\cal R}_C$.
    If $S$ is a $\cal F$-deletion set of size at most $k$, then $|N[C] \setminus S_U| \leq 3(\eta + 1)$ by Proposition~\ref{prop:pfd-modulator} and by Observation~\ref{obs:full-surgery} $\bigoplus_{\text{component }C} G[(X \cup Z)\setminus S] \oplus H_C$ is $\cal F$-minor free. $S$ therefore contributes $1$ to $\ccount(k')$. Further, if $|S| \leq k'$ then it also satisfies $|S_U| + \sum_C i_C \leq k'$.
    It is straightforward to verify each $\cal F$-deletion set of size at most $k'$ contributes $1$ to this sum.
    Further, if a set $S$ is not a $\cal F$-deletion set of size at most $k'$, then one of the above three statements is false and it contributes $0$ to this sum. Hence, the lemma holds.
\end{proof}

\begin{observation}\label{obs:count-time}
    $\ccount(k')$ can be evaluated in $2^{\OO(k^{3 + 9(\eta + 1)} \log k)} \cdot n^{\OO(1)}$ time.
\end{observation}
\begin{proof}
    The time required is more precisely expressed as $k^{3k} \cdot (c_{3(\eta + 1),\cal F} \cdot k)^{\tau_{\#comp}} \cdot n^{\OO(1)}$.
    The time is calculated by simply considering all possible choices of $S_U$, the collections $\{R_C\}$ and $\{i_C\}$. For each choice we can test the required conditions in polynomial time, and then take a product of the values picked from the tables $\{T_C\}$ in polynomial time.
\end{proof}

\subsubsection{The {\sf reduce} procedure and compression}

Let us now describe the {\sf reduce} procedure. The input is $(G,k)$ where $G$ is a graph on $n$ vertices.
We fix $\tau_n = 20(\eta + 1)$, and note that it is a constant that depends only on $\cal F$.
We first check if $\log n \leq k^{\tau_n}$.
If not, then we simply output the empty set $\phi$.
Otherwise, $\log n \leq k^{\tau n}$. We then apply Proposition~\ref{prop:pfd-modulator} to obtain the enhanced modulator $X \cup Z$.
Either it returns that $G$ has no $\cal F$-deletion set of size $k$, in which case we output $\phi$. Otherwise, we obtain an enriched modulator $X \cup Z$ of size $\OO(k^3)$. Then we apply Reduction Rule~\ref{rr:pfd-irr} exhaustively to find and delete components $C$ of $G - (X\cup Z)$ such that every pair $(B,H)$ realized by $C$ is rich. Recall that by Lemma~\ref{lemma:pdf-irr-comp}, each such component is disjoint from any minimal solution of $G$ of size at most $k$. Let $G'$ be the resulting graph, and note that each component of $G' - (X \cup Z)$ is also a component of $G - (X \cup Z)$. The next step is to compute the tables $T_C$ for each component $C$ of $G' - (X \cup Z)$. Here we apply Lemma~\ref{lemma:table} for each $C$ and obtain the table $T_C$ in polynomial time. The output of the reduce procedure is $(X \cup Z)$ and the table collection $\{ T_C \}$. Observe that each table requires at most $k^{\tau_n + 1} \cdot \tau_{table}$ bits of space. Since there are at most $\tau_{\#comp}$ components after an exhaustive application of Reduction Rule~\ref{rr:pfd-irr}, we have the following lemma.

\begin{lemma}\label{lemma:comp-size}
    Given instance $(G,k)$ of {\sc Planar $\cal F$-Deletion}, the {\sf reduce} procedure runs in polynomial time and outputs a data-structure of size $k^{\tau_n + 1} \cdot \tau_{table} \cdot \tau_{\#comp}$ which is a  polynomial in $k$.
\end{lemma}

\subsection{The {\sf lift} procedure}

The lift procedure is given the instance $(G,k)$, the output of the {\sf reduce} procedure, and for each $k' \leq k$ the value $\ccount(k')$ if the output of the {\sf reduce} procedure is not $\emptyset$.
Note that, when the output of the {\sf reduce} procedure is not $\emptyset$, then it consists of an enriched modulator $X \cup Z$ and a collection of tables $T_C$ for each component of $G' - (X \cup Z)$.
The objective is to the total number of $\cal F$-deletion sets in $G$ of size at most $k$.

The {\sf lift} procedure begins by applying Proposition~\ref{prop:pfd-modulator} to $(G,k)$. If it returns that $G$ has no $\cal F$-deletion set of size $k$, then we set $\tau_{total} = 0$ and output this value. Otherwise we have two cases, depending on whether$k^{\tau_n} \geq \log n$ or not.

First consider the case $k^{\tau_n} \geq \log n$. In this case, the {\sf reduce} procedure has output an enriched modulator $(X \cup Z)$ and a collection of tables $\{ T_C \}$, one for each non-irrelevant component of $G - (X\cup Z)$. We compute the total number of vertices in all the irrelevant components of $G - (X \cup Z)$ that were deleted by Reduction Rule~\ref{rr:pfd-irr}; it is denoted by $\tau_{irr}$; this can be computed in polynomial time by simulating the application of Reduction Rule~\ref{rr:pfd-irr}.
Recall that $G'$ denotes the graph obtained from $G$ after eliminating all irrelevant components.
Now, as the output of {\sf reduce} is not $\emptyset$, we are also given the values of $\ccount(k')$ for each $k' \leq k$ as a part of the input.
Recall $\ccount(k')$ denotes the total number of $\cal F$-deletion sets of size at most $k'$ in the graph $G'$. Let $\tau'_{count}(k') = \ccount(k') - \ccount(k'-1)$ denote the number of $\cal F$-deletion sets of size exactly $k'$ in $G'$.
Then we output $\tau_{total} = \sum_{k' \leq k} \tau'_{count}(k') \cdot {\tau_{irr} \choose k-k')}$ as the total number of solutions of size exactly $k$.

The other case is when $\log n > k^{\tau_n}$. In this case we proceed as follows. As in the previous case, we first apply Reduction Rule~\ref{rr:pfd-irr} to eliminate all the components of $G - (X \cup Z)$ that have no intersection with any minimal $\cal F$-deletion set of size at most $k$ in $G$. Let $\tau_{irr}$ denote the total number of vertices in the irrelevant components, and let $G'$ denote the remaining graph. The main difference from the above case is that, here we compute the value of $\ccount(k')$, for each $k' \leq k$, using the formula stated earlier. In this computation, instead of using the value of $T_C[{\cal R}_C, i_C, N[C] \setminus S_U]$ from the table $T_C$, we directly apply Lemma~\ref{lemma:table} to compute the value. Recall that each application of Lemma~\ref{lemma:table} takes polynomial time. Finally, as in the previous case we compute $\tau_{total}$, the number of $\cal F$-deletion sets of size at most $k$ in $G$, and output it.

\begin{lemma}
    $\tau_{total}$ is the total number of all $\cal F$-deletion set of size at most $k$ in $G$, and it is computed in polynomial time.
\end{lemma}
\begin{proof}
    If Proposition~\ref{prop:pfd-modulator} returns that $G$ has no $\cal F$-deletion set of size $k$, then clearly $\tau_{total} = 0$ is the correct answer. Otherwise, we compute $\tau_{total}$ using the values of $\ccount(k')$ for $k' \leq k$. It is easy to verify that every $\cal F$-deletion set of $G$ is counted in the formula for $\tau_{total}$.
    To account for the running time, the case in which $k^{\tau_n} \leq \log n$, the time taken is clearly polynomial.

    In the case $k^{\tau_n} > \log n$, that is, $2^{k^{20(\eta + 1)}} < n$, the only change is that the values of $\ccount(k')$ is computed by the {\sf lift} procedure directly, rather than being supplied externally as in the other case. We also apply Lemma~\ref{lemma:table} to compute the values $T_C[{\cal R}_C, i_C, N(C) \setminus S_U]$ for each choice of ${\cal R}_C, i_C$ and $N(C)\setminus S$. Recall that the total number of calls made to Lemma~\ref{lemma:table} is at most $\tau_{\#comp} \cdot \tau_{table}$ which is a polynomial in $k$, and each application takes polynomial in $n$ time. Next, we consider the computation of $\ccount(k')$ for some $k' \leq k$. By Observation~\ref{obs:count-time}, we need $2^{\OO(k^{3 + 9(\eta + 1)} \log k)} \cdot n^{\OO(1)}$ time for each evaluation of $\ccount(k')$. But as $n > 2^{k^{20(\eta + 1)}}$, it follows that the time required for each evaluation is $\OO(n^2)$. Thus in this case the {\sf lift} procedure requires polynomial time.
    %
\end{proof}

The {\sf reduce} and {\sf lift} procedures described above prove Theorem~\ref{thm:compression}.

\section{Lower Bounds Based on SUM-Cross-Composition}\label{sec:lowerSum}

We define two new notions of cross-compositions, which are suitable for parameterized counting problems. The first notion is defined as follows, and the second notion is defined in Section \ref{sec:lower}.

\begin{definition}[{\bf SUM-Cross-Composition}]\label{def:sumCross}
Let $P: \Sigma^\star\rightarrow\mathbb{N}_0$ be a counting problem and $Q : \Sigma^\star \times\mathbb{N}_0\rightarrow\mathbb{N}_0$ be a parameterized counting problem. We say that $P$ {\em SUM-cross-composes} into $Q$ if there exists a polynomial equivalence relation $R$ and an algorithm $A$, called a {\em SUM-cross-composition}, satisfying the following conditions. The algorithm $A$ takes as input a sequence of strings $x_1,x_2,\ldots,x_t \in\Sigma^\star$ that are equivalent with respect to $R$, runs in time polynomial in $\sum_{i=1}^t |x_i|$, and outputs one instance $(y, k)\in \Sigma^\star \times\mathbb{N}_0$ such that:
\begin{itemize}
\item  $k \leq p(\max^t_{i=1} |x_i| + \log t)$ for some polynomial function $p$, and
\item $Q(y,k)=\sum^t_{i=1}P(x_i)$.
\end{itemize}
\end{definition}

We pose the following conjecture, which will be the basis of the lower bounds presented in this section.

\begin{conjecture}[{\bf SUM-Conjecture}]
Assume that a \#P-hard counting problem $P$ SUM-cross-composes into a well-behaved parameterized counting problem $Q$. Then, $Q$ does not admit a polynomial compression.
\end{conjecture}

We first analyze the {\sc \#$k$-Min $(s,t)$-Cut} problem, whose unparameterized version is \#P-hard:

\begin{proposition}[\cite{provan1983complexity}]\label{prop:minCutHard}
{\sc \#Min $(s,t)$-Cut} is \#P-hard.
\end{proposition}

Since {\sc \#Min $(s,t)$-Cut} is \#P-hard, we will derive the hardness of the kernelization of {\sc \#$k$-Min $(s,t)$-Cut} from the following lemma.

\begin{restatable}{lemma}{stCutSum}
\label{lem:stCutSum}
{\sc \#Min $(s,t)$-Cut} SUM-cross-composes into {\sc \#$k$-Min $(s,t)$-Cut}.
\end{restatable}

\begin{proof}
First, we specify the equivalence relation $R$: Two strings $x$ and $x'$ satisfy $x\equiv_R x'$ if and only if they do not encode instances of {\sc \#Min $(s,t)$-Cut}, or they encode instances $x=(G,s,t)$ and $x'=(G',s',t')$ of {\sc \#Min $(s,t)$-Cut} and the size of a minimum $(s,t)$-cut in $G$ is equal to the size of a minimum $(s',t')$-cut in $G'$. Because the {\sc Min $(s,t)$-Cut} problem is solvable in polynomial time~\cite{cormen2001introduction}, it follows that $R$ is polynomial.

Now, we describe the SUM-cross-composition. For this purpose, consider a sequence of strings $x_1,x_2,\ldots,x_\ell$ that are equivalent with respect to $R$. If they do not encode instances of {\sc \#Min $(s,t)$-Cut}, then we can simply output a string that does not encode an instance of {\sc \#$k$-Min $(s,t)$-Cut}. Hence, we suppose that for every $i\in[\ell]$, $x_i=(G_i,s_i,t_i)$; then, the size of a minimum $(s_i,t_i)$-cut in $G_i$ is $k$. Without loss of generality, we suppose that the vertices of these graphs are distinct (otherwise, we can rename them). Now, we construct an instance $(G,s,t)$ of {\sc \#$k'$-Min $(s,t)$-Cut}. (We will argue that $k'=k$.) Let: \[V(G)=\bigcup_{i=1}^\ell(V(G_i)\setminus\{t_i\})\cup\{t_\ell\},\ \mathrm{and}\] 
\[E(G)=\bigcup_{i=1}^{\ell-1} (E(G_i-\{t_i\})\cup\{\{v,s_{i+1}\}: \{v,t_i\}\in E(G_i)\})\cup E(G_\ell).\]
Let $s=s_1$ and $t=t_\ell$. We refer to Fig.~\ref{fig:cut1} for an illustration.  Clearly, the construction can be done in polynomial time. For the sake of simplicity of the presentation, for every $i\in[\ell-1]$, we abuse notation and refer to the vertex $s_{i+1}$ in $G$ also as $t_{i}$; thus, for example, we refer to an edge $\{u,s_{i+1}\}$ in $G$ also as the edge $\{u,t_{i}\}$ (which belongs to $G_{i}$). Observe that, under this notation abuse, we simply have that $E(G)=\bigcup_{i=1}^\ell E(G_i)$. 

\begin{figure}
\begin{center}
  \includegraphics[scale=0.8]{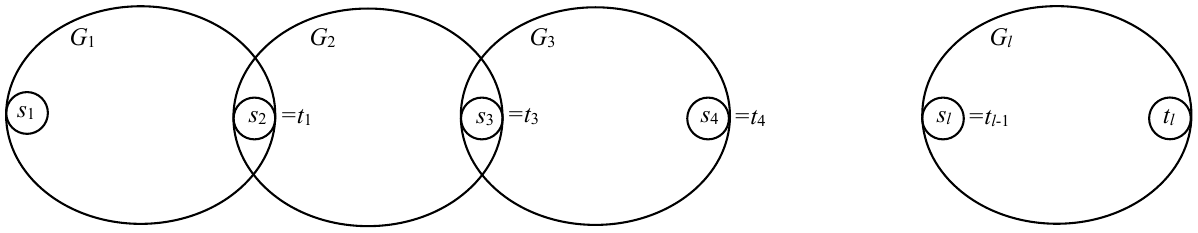}
  \caption{The construction in the proof of Lemma \ref{lem:stCutSum}.}
  \label{fig:cut1}
  \end{center}
\end{figure}

For the correctness of the composition, we state the two following claims. The correctness of these two claims is immediate from the construction of $G$.

\begin{claim}\label{claim:st1}
Let $S$ be a minimum $(s,t)$-cut in $G$. Then, there exists $i\in[\ell]$ such that $S\subseteq E(G_i)$ and $S$ is a minimum $(s_i,t_i)$-cut in $G_i$.
\end{claim}

\begin{claim}\label{claim:st2}
Let $S$ be a minimum $(s_i,t_i)$-cut in $G_i$, for some $i\in[\ell]$. Then, $S$ is a minimum $(s,t)$-cut in $G$.
\end{claim}

Observe that, from Claims \ref{claim:st1} and \ref{claim:st2}, it follows that $k'=k$, and that the number of minimum $(s,t)$-cuts in $G$ is equal to the sum, over all $i\in[\ell]$, of the number of minimum $(s_i,t_i)$-cuts in $G_i$. Thus, the proof is complete.
\end{proof}

Having Lemma \ref{lem:stCutSum} at hand, we proceed to consider PPTs that transfer the hardness to {\sc \#$k$-Odd Cycle Transversal} and {\sc \#$\ell$-Vertex Cover} (and {\sc \#$m$-Vertex Cover}). First, we present a PPT from {\sc \#$k$-Min $(s,t)$-Cut} to  {\sc \#$k$-Odd Cycle Transversal}. For the PPT from {\sc \#$k$-Odd Cycle Transversal} to {\sc \#$\ell$-Vertex Cover} (and {\sc \#$m$-Vertex Cover}), we will suppose that the instances of {\sc \#$k$-Odd Cycle Transversal} satisfy a particular property, hence we already define it now, and prove that our PPT from {\sc \#$k$-Min $(s,t)$-Cut} to  {\sc \#$k$-Odd Cycle Transversal} only produces instances with this property.

\begin{definition}[{\bf Nice Instances of {\sc \#$k$-Odd Cycle Transversal}}]
An instance $(G,k)$ of {\sc \#$k$-Odd Cycle Transversal} is {\em nice} if for every odd cycle transversal $S$ of $G$ of size at most $k$, $G-S$ is a connected graph.
\end{definition}

We now present our PPT from {\sc \#$k$-Min $(s,t)$-Cut} to  {\sc \#$k$-Odd Cycle Transversal}.

\begin{lemma}\label{lem:PPTOCT}
There exists a PPT from {\sc \#$k$-Min $(s,t)$-Cut} to {\sc \#$k$-Odd Cycle Transversal}. Moreover, the PPT only produces nice instances of {\sc \#$k$-Odd Cycle Transversal}.
\end{lemma}

\begin{proof}
For the description of the PPT, let $(G,s,t)$ be an instance of {\sc \#$k$-Min $(s,t)$-Cut}. Without loss of generality, we suppose that $G$ is a connected graph, else we can discard all connected components that do not contain $s$ or $t$, and, if $s$ and $t$ are not in the same connected component, then we already know that the solution is $1$ (the only minimum $(s,t)$-cut is the empty set), and hence the PPT is trivial. Then, we construct an instance $(G',k)$ of {\sc \#$k$-Odd Cycle Transversal}. Here, the parameter $k$ is the size of a minimum $(s,t)$-cut in $G$.

First, let $G_1$ be the graph obtained from $G$ by subdividing each edge (once). For an edge $\{u,v\}\in E(G)$, we denote the corresponding vertex in $G_1$ by $a_{\{u,v\}}$.
Let $G_2$ be the graph whose vertex set is $\{v_i: v\in V(G), i\in [k+1]\}\cup(V(G_1)\setminus V(G))$, and whose edge set is $\{\{u_i,a_{\{u,v\}}\}: \{u,v\}\in E(G), i\in[k+1]\}$. That is, $G_2$ is the result of the replacement of every vertex of $G_1$ that belongs to $G$ by $k+1$ copies (false twins) of that vertex.
Lastly, we define $G'$:
\[V(G')=V(G_2)\cup\{x_i: i\in[k+1]\}\cup\{y_i: i\in[k+1]\},\ \mathrm{and}\]
\[E(G')=E(G_2)\cup \{\{x_i,y_i\}:i\in[k+1]\}\cup\{\{s_i,x_j\}: i,j\in[k+1]\}\cup \{\{t_i,x_j\}: i,j\in[k+1]\}.\]
Clearly, the construction (performed by the reduction procedure of the PPT) can be done in polynomial time.

For correctness, we have the following claims.

\begin{claim}\label{claim:bipartite}
Let $G''=G'-\{\{x_i,y_i\}: i\in[k+1]\}$. Then, $G''$ is bipartite.
\end{claim}

\begin{proof}
Consider the following partition $(X,Y)$ of $V(G'')$:
\[X=\{a_{\{u,v\}}: \{u,v\}\in E(G)\}\cup\{x_i: i\in[k+1]\}\cup\{y_i: i\in[k+1]\},\ \mathrm{and}\]
\[Y=V(G'')\setminus X=\{v_i: v\in V(G), i\in[k+1]\}.\]
From the construction of $E(G'')$, it should be clear that $E(G'')\subseteq\{\{x,y\}: x\in X, y\in Y\}$.
\end{proof}

\begin{claim}\label{claim:oct1}
Let $C$ be an odd cycle in $G'$. Then, there exists a path $P=v^1-v^2-\ldots-v^\ell$ where $v^1=s$ and $v^\ell=t$ in $G$, and $i_1,i_2,\ldots,i_\ell\in[k+1]$ such that: \[C\supseteq v^1_{i_1}-a_{\{v^1,v^2\}}-v^2_{i_2}-a_{\{v^2,v^3\}}-v^3_{i_3}-\cdots - v^{\ell-1}_{i_{\ell-1}}-a_{\{v^{\ell-1},v^\ell\}}-v^\ell_{i_\ell}.\]
\end{claim}

\begin{proof}
Due to Claim \ref{claim:bipartite} and since a graph is bipartite if and only if it does not contain any odd cycle, there exist $j\in[k+1]$ such that $\{x_j,y_j\}\subseteq E(C)$. Targeting a contradiction, suppose that $C$ does not contain a path of the form stated in the lemma. Thus, the definition of $G'$ implies that there exist $r_1,r_2,\ldots,r_s$ such that
\[C=x_{r_0}-P_{r_1}-x_{r_1}-y_{r_1}-P_{r_2}-y_{r_2}-x_{r_2}-P_{r_3}-x_{r_3}-y_{r_3}-\cdots-P_{r_{s-1}}-x_{r_{s-1}}-y_{r_{s-1}}-P_{r_s}-y_{r_s}-x_{r_s},\]
where $x_j=x_{r_0}=x_{r_s}$ and $y_j=y_{r_s}$, and, for every $i\in[s]$, $P_{r_i}$ is a path in $G''$ (defined in Claim \ref{claim:bipartite}) whose endpoints satisfy that they are adjacent to the vertices specified above ($x_{r_{i-1}}$ and $x_{r_i}$ if $i$ is odd, and $y_{r_{i-1}}$ and $y_{r_i}$ if $i$ is even). Observe that, necessarily, $s$ is even. From Claim \ref{claim:bipartite} (specifically, consider the bipartition defined in the proof), we know that each $P_{r_i}$, $i\in[s]$, along with the edge before it and the edge after it, has an even number of edges. Besides this, all other edges of $C$ are $\{x_{r_i},y_{r_i}\}$, $i\in[s]$. However, since $s$ is even, this means that their number is even as well. Overall, we derive that $C$ contains an even number of edges, which is a contradiction (since $C$ is an odd cycle).
\end{proof}

\begin{claim}\label{claim:oct2}
Let $P=v^1-v^2-\ldots-v^\ell$, where $v^1=s$ and $v^\ell=t$, be a path in $G$, and let $i_1,i_2,\ldots,i_\ell,j\in[k+1]$. Additionally, let: \[C=v^1_{i_1}-a_{\{v^1,v^2\}}-v^2_{i_2}-a_{\{v^2,v^3\}}-v^3_{i_3}-\cdots - v^{\ell-1}_{i_{\ell-1}}-a_{\{v^{\ell-1},v^\ell\}}-v^\ell_{i_\ell}-x_j-y_j-v^1_{i_1}.\]
Then, $C$ is an odd cycle in $G'$.
\end{claim}

\begin{proof}
From the definition of $G'$, it is immediate that $C$ is a cycle in $G'$, and, clearly, $C$ contains an odd number (being $2(\ell-1)+3=2\ell+1$) of edges.
\end{proof}

\begin{claim}\label{claim:MINCUTtoOCT}
Let $S$ be a minimum $(s,t)$-cut in $G$. Then, $S'=\{a_e: e\in S\}$ is an odd cycle transversal of size at most $k$ in $G'$.
\end{claim}

\begin{proof}
Since $|S|=k$, it follows that $|S'|\leq k$. Now, targeting a contradiction, suppose that $S'$ is not an odd cycle transversal of $G'$. So, $G'-S'$ contains some odd cycle cycle $C$. By Claim \ref{claim:oct1}, there exists a path $P=v^1-v^2-\ldots-v^\ell$ where $v^1=s$ and $v^\ell=t$ in $G$, and $i_1,i_2,\ldots,i_\ell\in[k+1]$ such that $C\supseteq v^1_{i_1}-a_{\{v^1,v^2\}}-v^2_{i_2}-a_{\{v^2,v^3\}}-v^3_{i_3}-\cdots - v^{\ell-1}_{i_\ell}-a_{\{v^{\ell-1},v^\ell\}}-v^\ell_{i_\ell}$. Since $V(C)\cap S'=\emptyset$, we have that $a_{\{v^1,v^2\}},a_{\{v^2,v^2\}},\ldots,a_{\{v^{\ell-1},v^{\ell}\}}\notin S$. In turn, this implies that $P$ exists in $G-S$. However, since $S$ is an $(s,t)$-cut in $G$, we have thus reached a contradiction.
\end{proof}

\begin{claim}\label{claim:OCTtoMINCUT}
Let $S'$ be an odd cycle cycle transversal of $G'$ of size at most $k$. Then, $S=\{e\in E(G): a_e\in S\}$ is a minimum $(s,t)$-cut in $G$. Moreover, $G'-S'$ is a connected graph.
\end{claim}

\begin{proof}
We first show that $S$ is a minimum $(s,t)$-cut in $G$. Since $|S'|\leq k$, it follows that $|S|\leq k$. Now, targeting a contradiction, suppose that $S$ is not an $(s,t)$-cut in $G$. So, $G-S$ contains some path $P=v^1-v^2-\ldots-v^\ell$ where $v^1=s$ and $v^\ell=t$. Because $|S'|\leq k$, there exist $i_1,i_2,\ldots,i_\ell,j\in[k+1]$ such that, for every $r\in[\ell]$, $v^r_{i_r}\notin S'$, and $x_j,y_j\notin S'$. By Claim \ref{claim:oct2}, $C=v^1_{i_1}-a_{\{v^1,v^2\}}-v^2_{i_2}-a_{\{v^2,v^3\}}-v^3_{i_3}-\cdots - v^{\ell-1}_{i_\ell}-a_{\{v^{\ell-1},v^\ell\}}-v^\ell_{i_\ell}-x_j-y_j-v^1_{i_1}$ is an odd cycle in $G'$. Since $P$ belongs to $G-S$, we have that $a_{\{v^1,v^2\}}, a_{\{v^2,v^3\}},\ldots,a_{\{v^{\ell-1},v^\ell\}}\notin S'$. Hence, from our choice of $i_1,i_2,\ldots,i_\ell,j$, it follows that $C$ belongs to $G'-S'$. However, since $S'$ is an odd cycle transversal of $G'$, we have thus reached a contradiction.

Because $S$ is a minimum $(s,t)$-cut in $G$, it follows that $|S|=k$. So, $S'=\{a_e: e\in S\}$. Moreover, because $G$ is a connected graph and $S$ is a minimum $(s,t)$-cut in $G$, $G-S$ consists of exactly  two connected components: one component that contains $s$, and the other component that contains $t$. However, by the definition of $G_2$, this implies that in $G_2-S'$, every connected component must contain $s_i$ or $t_i$ for some $i\in[k+1]$. In turn, by the definition of $G'$, this implies that $G'-S'$ is a connected graph.
\end{proof}

Observe that, from Claims \ref{claim:MINCUTtoOCT} and \ref{claim:OCTtoMINCUT}, it follows that the number of minimum $(s,t)$-cuts in $G$ is equal to the number of odd cycle transversals of $G'$ of size at most $k$. (In fact, every odd cycle transversal of $G'$ of size at most $k$ is of size exactly $k$.)  Moreover, the second part of Claim \ref{claim:OCTtoMINCUT} shows that $(G',k)$ is nice. So, given the the number of odd cycle transversals of $G'$ of size at most $k$, the lifting procedure of the PPT simply outputs this number. Thus, the proof is complete.
\end{proof}

\begin{restatable}{lemma}{PPTVC}
\label{lem:PPTVC}
There exists a PPT from {\sc \#$k$-Odd Cycle Transversal} restricted to nice instances to {\sc \#$\ell$-Vertex Cover} and {\sc \#$m$-Vertex Cover}.
\end{restatable}

\begin{proof}
We refer to the PPT in \cite{DBLP:books/sp/CyganFKLMPPS15} (see Lemma 3.10) from {\sc $k$-Odd Cycle Transversal} to {\sc $\ell$-Vertex Cover} and {\sc $m$-Vertex Cover}. From the construction (and the proof of Lemma 3.10), {\em and because we only deal with nice instances}, we can see that, given an instance $(G,k)$ of {\sc $k$-Odd Cycle Transversal}, and the produced instance $(G',k')$ of {\sc $\ell$-Vertex Cover} (or {\sc $m$-Vertex Cover}), the number of odd cycle transversals of $G$ of size at most $k$ is exactly half the number of vertex covers of $G'$ of size at most $k'$. Thus, the correctness of the lemma follows. For the sake of completeness, we present the details in Appendix \ref{sec:vcProof}.
\end{proof}

Observe that all problems considered in this section are well-behaved: for a graph on $n$ and $m$ vertices and edges, $2^{n+m}$ is a trivial upper bound on the number of solutions for all of these problems. So, from Lemmas \ref{lem:countingPPT}, \ref{lem:stCutSum}, \ref{lem:PPTOCT} and \ref{lem:PPTVC}, we directly conclude the following theorem.

\lowerSUM*

\section{Lower Bound Based on EXACT-Cross-Composition}\label{sec:lower}

Our second new notion of a cross-composition is defined as follows. 

\begin{definition}[{\bf EXACT-Cross-Composition}]\label{def:exactCross}
Let $P: \Sigma^\star\rightarrow\mathbb{N}_0$ be a counting problem and $Q : \Sigma^\star \times\mathbb{N}_0\rightarrow\mathbb{N}_0$ be a parameterized counting problem. We say that $P$ {\em EXACT-cross-composes} into $Q$ if there exists a polynomial equivalence relation $R$ and an algorithm $A$, called an {\em EXACT-cross-composition}, satisfying the following conditions. The algorithm $A$ takes as input a sequence of strings $x_1,x_2,\ldots,x_t \in\Sigma^\star$ that are equivalent with respect to $R$, runs in time polynomial in $\sum_{i=1}^t |x_i|$, and outputs one instance $(y, k)\in \Sigma^\star \times\mathbb{N}_0$ such that:
\begin{itemize}
\item  $k \leq p(\max^t_{i=1} |x_i| + \log t)$ for some polynomial function $p$, and
\item there exists a polynomial-time procedure that, given $x_1,x_2,\ldots,x_t ,(y,k)$ and $Q(y,k)$, outputs $P(x_1),P(x_2),\ldots,P(x_t)$.
\end{itemize}
\end{definition}

We remark that EXACT-cross-compositions seem to be harder to devise than SUM-cross-compositions. In particular, for EXACT-cross-compositions, but not for SUM-cross-compositions, we are able to prove the following theorem.

\lowerThmEXACT*

By \#P $\subseteq$ ``NP/poly'', we mean that, for  any \#P-complete problem $P$, there exists a nondeterministic polynomial-time algorithm $A$  and a sequence of strings $(\alpha_n)_{n=0,1,2,\ldots}$, called {\em advice}, such that:
\begin{enumerate}
\item Given an instance $x$ of $P$ of size $n$, $A$ has access to $\alpha_n$, and:
	 \begin{enumerate}
	 \item For every computation path of $A$, the output is either $P(x)$ or ``Do Not Know''.
	 \item There exists a computation path of $A$ whose output is $P(x)$.
	 \end{enumerate}
\item There exists a polynomial $p:\mathbb{N}\rightarrow\mathbb{N}$ such that $|\alpha_i|\leq p(i)$ for every $i\in\mathbb{N}$.
\end{enumerate}

Another way to think of the phrase ``NP/poly'' is as follows. Observe that we can ``force'' a counting problem $P$ to be a decision problem $P'$ by the addition, to each of its instances, of another argument $s$, and, accordingly, modifying its task to that of deciding whether $P(x)\geq s$. Further, if we let $N_n$ be the maximum bitsize of the encoding of $P(x)$ for any instance $x$ of $P$ of size $n$ (that is polynomially bounded, since $P$ is well-behaved), then we can solve $P$ itself by making polynomially many calls to an algorithm for $P'$: We use this algorithm to perform a binary search on the range $[N_n]$. Under this interpretation, our proof implies the following statement. If a \#P-hard problem $P$ EXACT-cross-composes into a parameterized counting problem $Q$, then $Q$ does not admit a polynomial compression unless the ``forced'' decision version $P'$ of $P$ can be solved by a (standard) NP/poly algorithm, or, alternatively, $P$ can be solved by making polynomially many calls to a (standard) NP/poly algorithm.

Additionally, we would like to point out that the supposition that coNP is not contained in NP/poly (which is widely believed to be true, and it is the standard supposition on which hardness results for kernelization algorithms are based~\cite{fomin2019kernelization}) implies the supposition that \#P is not contained in ``NP/poly''. To see this, suppose that \#P $\subseteq$ ``NP/poly''. For example, this implies that {\sc \#$k$-Vertex Cover}, which is \#P-hard~\cite{DBLP:journals/cc/Greenhill00}, belongs ``NP/poly''. Now, consider the complement of {\sc $k$-Vertex Cover}, denoted by $Q$: Given a graph $G$ and a non-negative integer $k$, decide whether all vertex covers of $G$ are of size larger than $k$. Since {\sc $k$-Vertex Cover} is NP-hard~\cite{karp1972reducibility}, $Q$ is coNP-hard. However, because we suppose that {\sc \#$k$-Vertex Cover} belongs to ``NP/poly'', the above discussion implies we can determine whether its solution is at least one by making a single call to a (standard) NP/poly algorithm. However, this solves $Q$, and, hence, we derive that coNP $\subseteq$ NP/poly.

\begin{proof}[Proof of Theorem \ref{thm:exact}]
The proof follows lines similar to that of Proposition \ref{prop:crossComp}. For the sake of completeness, we present the details in Appendix \ref{sec:exactProof}.
\end{proof}

We proceed to present an EXACT-cross-composition for {\sc \#$w$-Min $(s,t)$-Cut}. We remark that we do not know how to present EXACT-cross-composition for the problems in Section \ref{sec:lowerSum}.

\begin{lemma}\label{lem:lowerExact}
{\sc \#Min $(s,t)$-Cut} EXACT-cross-composes into {\sc \#$w$-Min $(s,t)$-Cut}.
\end{lemma}

\begin{proof}
The equivalence relation $R$ is the same as the one in the proof of Lemma \ref{lem:stCutSum}. Now, we describe the EXACT-cross-composition. For this purpose, consider a sequence of strings $x_1,x_2,\ldots,x_\ell$ that are equivalent with respect to $R$. Similarly to the proof of Lemma \ref{lem:stCutSum}, we suppose that for every $i\in[\ell]$, $x_i=(G_i,s_i,t_i)$; then, the size of a minimum $(s_i,t_i)$-cut in $G_i$ is $k$. Let $m'=\max_{i=1}^\ell|E(G_i)|$. For every $i\in[\ell]$, let $q_i$ denote the (unknown) number of minimum $(s_i,t_i)$-cuts in $G_i$ (which are of size $k$). Observe that, if $\ell\geq 2^{m'}$, then, in polynomial time, we can iterate over every subset of edges of each of the graphs $G_i,i\in[\ell]$, and thereby compute $q_1,q_2,\ldots,q_\ell$. In this case, the design of an EXACT-cross-composition is trivial, and hence we suppose that $\ell<2^{m'}$.

Now, we consider the construction of $(G,s,t)$ given in the proof of Lemma \ref{lem:stCutSum}. However, here, we modify $G$ further in order to attain the output instance. Let $m=2m'$. Then, the output instance of {\sc \#$w$-Min $(s,t)$-Cut} is $(G',s,t)$, where $G'$ is defined as follows:
\[V(G')=V(G)\cup\left(\bigcup_{i=1}^\ell\{x^j_i: j\in[m(\ell-1)]\}\cup\{y^j_i,z^j_i: j\in[m(i-1)]\}\right),\ \mathrm{and}\]
\[E(G')=E(G)\cup\left(\bigcup_{i=1}^\ell E_i\right),\ \mathrm{where}\]
\[E_i=\{\{s_i,x^j_i\},\{x^j_i,y^j_i\},\{y^j_i,z^j_i\},\{z^j_i,t_i\}: j\in[m(i-1)]\} \cup \{\{s_i,x^j_i\},\{x^j_i,t_i\}: j\in[m(\ell-1)]\setminus[m(i-1)]\}.\]
We refer to Fig.~\ref{fig:cut2} for an illustration.  Clearly, the construction can be done in polynomial time. 

\begin{figure}
\begin{center}
  \includegraphics[scale=0.8]{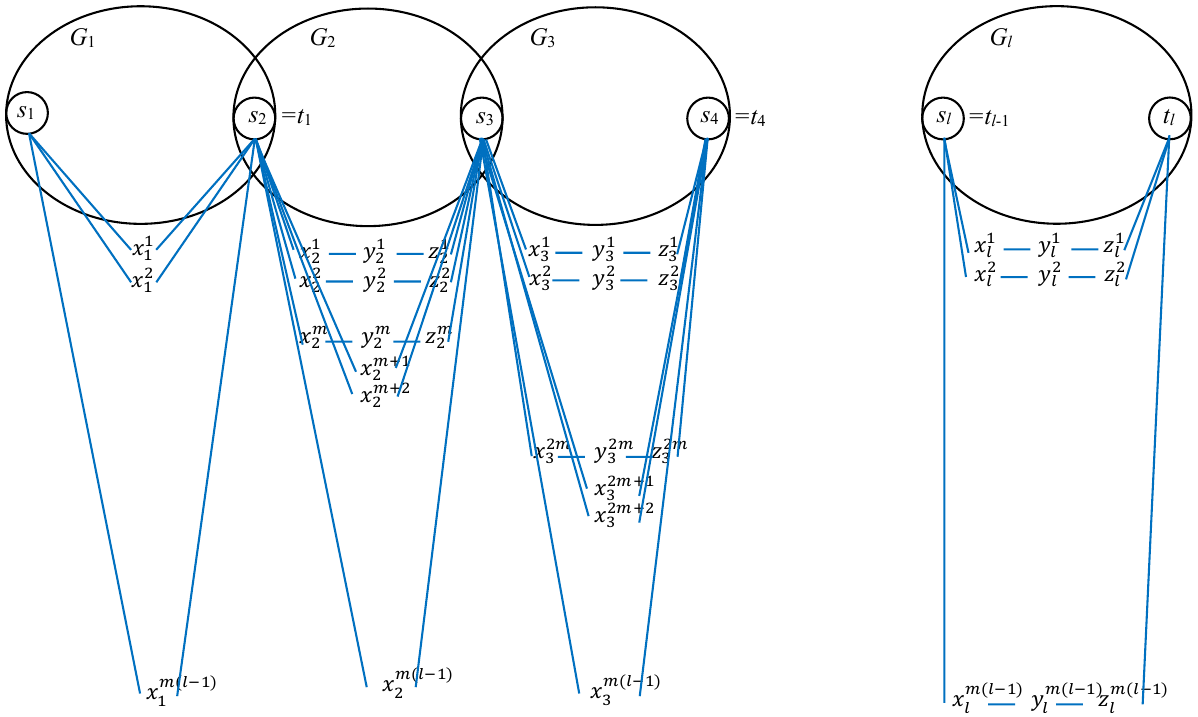}
  \caption{The construction in the proof of Lemma \ref{lem:stCutSum}.}
  \label{fig:cut2}
  \end{center}
\end{figure}

For the correctness of the composition, we first present an upper bound on the treewidth of $G'$, which is the parameter $w$ associated with $(G',s,t)$.

\begin{claim}\label{claim:tw}
$w=\mathsf{tw}(G')\leq\max\{2,\max_{i=1}^\ell\mathsf{tw}(G_i)+1\}$.
\end{claim}

\begin{proof}
Let $w^\star=\max_{i=1}^\ell\mathsf{tw}(G_i)$. For every $i\in[\ell]$, let ${\cal T}_i=(T_i,\beta_i)$ be a tree decomposition of $G_i$ of width $\mathsf{tw}(G_i)\leq w^\star$, and define ${\cal T}'_i=(T'_i,\beta'_i)$ as follows:
\begin{itemize}
\item Choose $v_i\in V(T_i)$ such that $s_i\in\beta(a_i)$.
\item $V(T'_i)=V(T_i)\cup \{a^j_i: j\in[m(\ell-1)]\}\cup\{b^j_i,c^j_i: j\in[m(i-1)]\}$.
\item $E(T'_i)=E(T_i)\cup \{\{v_i,a^j_i\}: j\in[m(\ell-1)]\}\cup\{\{a^j_i,b^j_i\},\{b^j_i,c^j_i\}: j\in [m(i-1)]\}$.
\item For every $u\in V(T_i)$, $\beta'_i(u)=\beta_i(u)\cup\{t_i\}$.
\item For every $j\in [m(\ell-1)]$,$\beta'_i(a^j_i)=\{s_i,x^j_i,t_i\}$.
\item For every $j\in [m(i-1)]$, $\beta'_i(b^j_i)=\{x^j_i,y^j_i,t_i\}$ and $\beta'_i(c^j_i)=\{y^j_i,z^j_i,t_i\}$.
\end{itemize}
It is straightforward to verify that ${\cal T}'_i$ is a tree decomposition of $G'[V(G_i)\cup \{x^j_i: j\in[m(\ell-1)]\}\cup\{y^j_i,z^j_i: j\in[m(i-1)]\}]$, and its width is $w_i'=\max\{2,\mathsf{tw}(G_i)+1\}\leq \max\{2,w^\star+1\}$.

Now, we define ${\cal T}'=(T',\beta')$ as follows:
\begin{itemize}
\item $V(T')=\bigcup_{i=1}^\ell V(T'_i)$.
\item $E(T')=\bigcup_{i=1}^{\ell-1} (E(T'_i)\cup\{\{v_i,v_{i+1}\}\})\cup E(T'_{\ell})$.
\item For every $i\in[\ell]$ and $u\in V(T'_i)$, $\beta'(u)=\beta'_i(u)$.
\end{itemize}
It is straightforward to verify that ${\cal T}'$ is a tree decomposition of $G'$, and its width is bounded from above by $\max\{2,w^\star+1\}$. This completes the proof of the claim.
\end{proof}

Second, we present two immediate claims about the correspondence between the cuts of $G_i$, $i\in[\ell]$, and the cuts of $G'$. Towards this, for all $i\in[\ell]$, let $C_i=\{(\{a_1,b_1\},\{a_2,b_2\},\ldots,\{a_{m(\ell-1)},b_{m(\ell-1)}\}):$ for all $j\in [m(i-1)]$, $\{a_j,b_j\}\in\{\{s_i,x^j_i\},\{x^j_i,y^j_i\},\{y^j_i,z^j_i\},\{z^j_i,t_i\}\}$, and for all $j\in[m(\ell-1)]\setminus[m(i-1)], \{a_j,b_j\}\in\{\{s_i,x^j_i\},\{x^j_i,t_i\}\}\}$; observe that $|C_i|=2^{m(i-1)}\cdot 2^{m(\ell-1)}$.

\begin{claim}\label{claim:mincut1}
Let $S$ be a minimum $(s,t)$-cut in $G$. Then, there exists $i\in[\ell]$ such that $S=A\cup B$ where $A$ is a minimum $(s,t)$-cut in $G_i$ and $B\in C_i$.
\end{claim}

\begin{claim}\label{claim:mincut2}
Let $A$ be a minimum $(s_i,t_i)$-cut in $G_i$, for some $i\in[\ell]$. Then, for all $B\in C_i$, $S=A\cup B$ is a minimum $(s,t)$-cut in $G$.
\end{claim}

Additionally, let $q$ denote the number of minimum $(s,t)$-cuts in $G'$. So, from Claims \ref{claim:mincut1} and \ref{claim:mincut2}, we arrive at the following conclusion.

\begin{claim}\label{claim:q}
$q=\sum_{i=1}^\ell q_i\cdot 2^{m(i-1)}\cdot 2^{m(\ell-1)}$.
\end{claim}

We are now ready to show how to extract each $q_i$, $i\in[\ell]$, given $q$.

\begin{claim}\label{claim:extractq}
There exists a polynomial-time procedure that, given $(G_1,s_1,t_1),(G_2,s_2,t_2),\ldots,(G_\ell,s_\ell,t_\ell),$ $(G',s,t)$ and $q$, outputs $q_1,q_2,\ldots,q_\ell$.
\end{claim}

\begin{proof}
The procedure performs the following operations:
\begin{enumerate}
\item Initialize $\widehat{q}\leftarrow q$.
\item For $i=\ell,\ell-1,\ldots,1$:
	\begin{enumerate}
	\item Let $q_i\leftarrow \lfloor\widehat{q}/(2^{m(i-1)}\cdot 2^{m(\ell-1)})\rfloor$.
	\item Update $\widehat{q}\leftarrow \widehat{q}-q_i\cdot 2^{m(i-1)}\cdot 2^{m(\ell-1)}$.
	\end{enumerate}
\item Return $q_1,q_2,\ldots,q_\ell$.
\end{enumerate}
Clearly, the procedure runs in polynomial time. Additionally, recall that $\ell<2^{m'}$ and $m=2m'$, and observe that for every $i\in[\ell]$, we have that $q_i\in[2^{m'}]$. Thus, for every $i\in[\ell]$, we have that \[\begin{array}{ll}
2^{m(i-1)}\cdot 2^{m(\ell-1)}& > \ell\cdot 2^{m'}\cdot 2^{m((i-1)-1)}\cdot 2^{m(\ell-1)}\\
& \geq \sum_{j=1}^{i-1}2^{m'}\cdot 2^{m((i-1)-1)}\cdot 2^{m(\ell-1)}\\
& \geq \sum_{j=1}^{i-1}q_j\cdot 2^{m(j-1)}\cdot 2^{m(\ell-1)}.
\end{array}\]
Due to this inequality, the correctness of the procedure follows from Claim \ref{claim:q}.
\end{proof}

Thus, the correctness of the composition follows from Claims \ref{claim:tw} and \ref{claim:extractq}.
\end{proof}

As already noted in the previous section, {\sc \#Min $(s,t)$-Cut} (and, hence, also any parameterized version of it) is well-behaved. So, from Proposition \ref{prop:minCutHard}, Theorem \ref{thm:exact} and Lemma \ref{lem:lowerExact}, we directly conclude the following theorem.

\lowerEXACT*

\bibliographystyle{siam}
\bibliography{references}

\appendix


\section{Proofs Omitted from Section \ref{sec:prelims}}\label{sec:prelimsProofs}

\countingFPTKernel*

\begin{proof}
In one direction, suppose that $P$ admits a kernel, and let $\mathsf{reduce}$ and $\mathsf{lift}$ be the procedures corresponding to it. Let $\cal F$ be a finite-time algorithm for $P$.  Then, we design an FPT algorithm $\cal A$ for $P$ as follows. Given an instance $(x,k)$ of $P$, $\cal A$ calls $\mathsf{reduce}$ on $(x,k)$ to obtain, in polynomial time, an instance $(x',k')$ of $P$ whose size is bounded from above by $f(k)$ for some computable function $f$. Then, $\cal A$ calls $\cal F$ on $(x',k')$ to obtain $s'=P(x',k')$ in time at most $g(|x'|+k')\leq g(f(k))$ for some computable function $g$. Lastly, $\cal A$ calls $\mathsf{lift}$ on $(x,k),(x',k'),s'$ to obtain, in polynomial time, $s=P(x,k)$. Overall, the running time is FPT, and hence the proof of this direction is complete.

In the second direction, suppose that $P$ admits an FPT algorithm $\cal A$, whose running time is bounded from above by $f(k)\cdot |x|^c$ for some computable function $f$ and a fixed constant $c$. We define $\mathsf{reduce}$ and $\mathsf{lift}$ as follows. Let $(x,k)$ be an instance of $P$. Then:
\begin{enumerate}
\item If $|x|\leq f(k)$, then the output of $\mathsf{reduce}$ on $(x,k)$ is $(x,k)$ itself. Observe that, in this case, the requirements concerning the running time and the size of the output are trivially satisfied. Given the input $(x,k)$, the output $(x,k)$ and $s=P(x,k)$, $\mathsf{lift}$ simply outputs $s$.  
\item If $|x|>f(k)$, then $\mathsf{reduce}$ outputs some arbitrary instance $(x',k')$ of constant-size.  So, the requirements concerning the running time and the size of the output are trivially satisfied.  Given $(x,k),(x',k')$ and $s'=P(x',k')$, $\mathsf{lift}$ calls $\cal A$ on $(x,k)$, and returns its output, which is $s=P(x,k)$. Then, the running time of $\mathsf{lift}$ is bounded from above by $\OO(f(k)\cdot |x|^c)\leq \OO(|x|^{c+1})$, and, hence, $\mathsf{lift}$ runs in polynomial time.
\end{enumerate}
This completes the proof.
\end{proof}

\countingPPT*

\begin{proof}
Suppose that $Q$ admits a polynomial compression into some parameterized counting problem $R$, and let $\mathsf{reduce}_Q$ and $\mathsf{lift}_Q$ be the procedures corresponding to it. Let $\mathsf{reduce}_{\mathrm{PPT}}$ and $\mathsf{lift}_{\mathrm{PPT}}$ denote the procedures corresponding to the PPT from $P$ to $Q$. We define $\mathsf{reduce}_P$ and $\mathsf{lift}_P$ as follows. Given an instance $(x,k)$ of $P$, $\mathsf{reduce}_P$ calls $\mathsf{reduce}_{\mathrm{PPT}}$ on $(x,k)$ to obtain an instance $(x',k')$ of $Q$. Then, $\mathsf{reduce}_P$ calls $\mathsf{reduce}_Q$ on $(x',k')$ to obtain an instance $(x'',k'')$ of $R$, and returns $(x'',k'')$. Clearly, the running time is polynomial. Further, since the size of $(x'',k'')$ is bounded from above by some polynomial in $k'$, and $k'$ is bounded from above by some polynomial in $k$, so we get that the size of $(x'',k'')$ is bounded from above by some polynomial in $k$.

Now, given $(x,k),(x'',k'')$, and $s''=R(x'',k'')$, $\mathsf{lift}_P$ calls $\mathsf{lift}_Q$ on $(x',k'),(x'',k''),s''$ to obtain $s'$ such that $s'=Q(x',k')$. Then, $\mathsf{lift}_P$ calls $\mathsf{lift}_{\mathrm{PPT}}$ on $(x,k),(x',k'),s'$ to obtain $s$ such that $s=P(x,k)$. Clearly, the running time is polynomial. This completes the proof.
\end{proof}

\section{Proof of Lemma \ref{lem:PPTVC}}\label{sec:vcProof}

\PPTVC*

\begin{proof}
For the description of the construction, let $(G,k)$ be a nice instance of {\sc \#$k$-Odd Cycle Transversal}. The output of the reduction procedure of the PPT is the instance $(G',k')$ of {\sc \#$\ell$-Vertex Cover} (or {\sc \#$m$-Vertex Cover}) defined as follows. Let $G_1$ and $G_2$ be two copies of $G$. For $i\in[2]$ and $v\in V(G)$, let $v_i$ be the copy of $v$ in $G_i$. Then, $V(G')=V(G_1)\cup V(G_2)$ and $E(G')=E(G_1)\cup E(G_2)\cup\{\{v_1,v_2\}: v\in V(G)\}$. Additionally, $k'=|V(G)|+k$. Clearly, the construction can be done in polynomial time.

We first consider the value of the parameter of $(G',k')$:

\begin{claim}
$\ell=m=k$.
\end{claim}

\begin{proof}
Since $G$ has a perfect matching (consisting of the edges in $\{\{v_1,v_2\}: v\in V(G)\}$), we have that $\mathsf{LP}_{\mathsf{VC}}(G)=\mu(G)=|V(G)|$. So, $\ell=m=k'-|V(G)|=k$.
\end{proof}

Now, let ${\cal S}$ be the set of odd cycle transversals of $G$ of size at most $k$, and let ${\cal S}'$ be the set of vertex covers of $G'$ of size at most $k'$. Observe that for any $S\in {\cal S}$, because $G-S$ is connected, we have that $G-S$ has exactly two bipartitions, $(A,V(G-S)\setminus A)$ and $(V(G-S)\setminus A,A)$, for some $A\subseteq V(G-S)$, which we denote by $A_S$; also, let $B_S=V(G-S)\setminus A_S$.
We define a function $\mathsf{map}: {\cal S}\rightarrow{\cal S}'$ as follows. For any $S\in {\cal S}$, let $\mathsf{map}=\{U^1_S,U^2_S\}$ where $U^1_S=\{v_1: v\in A_S\cup S\}\cup\{v_2: v\in B_S\cup S\}$ and $U^2_S=\{v_2: v\in A_S\cup S\}\cup\{v_1: v\in B_S\cup S\}$.

First, we assert that $U^1_S,U^2_S\in{\cal S}'$.

\begin{claim}\label{claim:vc1}
For every $S\in{\cal S}$, $\mathsf{map}(S)\subseteq{\cal S}'$. Moreover, for distinct $S,T\in{\cal S}$, $\mathsf{map}(S)\cap\mathsf{map}(T)=\emptyset$.
\end{claim}

\begin{proof}
Consider some $S\in{\cal S}$. For the first part of the claim, we only show that $U^1_S\in{\cal S}'$, since the proof that $U^2_S\in{\cal S}'$ is symmetric. Since $B_S$ is an independent set in $G$, we have that $\{v_1: v\in A_S\cup S\}\subseteq U^1_S$ covers all edges in $G_1$, and since $A_S$ is an independent set in $G$, we have that $\{v_2: v\in B_S\cup S\}\subseteq U^1_S$ covers all edges in $G_2$. Additionally, for every $v\in V(G)$, $v$ is in $S$, in $A_S$ or in $B_S$, and hence $\{v_1,v_2\}\cap U^1_S\neq\emptyset$. So, $U^1_S$ covers $E(G')$. Moreover, $|U^1_S|=|V(G)|+|S|\leq |V(G)|+k=k'$. Hence, $U^1_S\in{\cal S}'$.

The second part of the claim is immediate from the definition of $\mathsf{map}$.
\end{proof}

\begin{claim}\label{claim:vc2}
For every $S'\in{\cal S}'$, there exists $S\in{\cal S}$ such that $S'\in\mathsf{map}(S)$. (That is, $\mathsf{map}$ is surjective.)
\end{claim}

\begin{proof}
Consider some $S'\in{\cal S}'$.  Let $S=\{v\in V(G): v_1,v_2\in S'\}$. We claim that $S\in{\cal S}$. For this purpose, let $X_S=\{v\in V(G): v_1\notin S', v_2\in S'\}$ and $Y_S=\{v\in V(G): v_2\notin S', v_1\in S'\}$. Since $S'$ covers $E(G')$ and in particular $\{\{v_1,v_2\}: v\in V(G)\}$, we have that $(S,X_S,Y_S)$ is a partition of $V(G)$, and $|S'|=|V(G)|+|S|$. Moreover, since $S'$ covers $E(G_1)$ (resp., $E(G_2)$), $G[X_S]$ (resp., $G[Y_S]$) must be an independent set. So, $(X_S,Y_S)$ is a bipartition of $G-S$, which means that $S$ is an odd cycle transversal of $G$. Since $|S'|\leq k'=|V(G)|+k$, we get that $|S|\leq k$. Hence, $S\in{\cal S}$. Now, notice that either $A_S=X_S$ and $B_S=Y_S$ (then, $S'=U^2_S$) or $B_S=X_S$ and $A_S=Y_S$ (then, $S'=U^1_S$). So, $S'\in\mathsf{map}(S)$.
\end{proof}

Observe that, from Claims \ref{claim:vc1} and \ref{claim:vc2}, it follows $|{\cal S}|=\frac{1}{2}|{\cal S}'|$. So, given $|{\cal S}'|$, the lifting procedure of the PPT simply outputs $\frac{1}{2}|{\cal S}'|$. Thus, the proof is complete.
\end{proof}


\section{Proof of Theorem \ref{thm:exact}}\label{sec:exactProof}

Towards the proof of Theorem \ref{thm:exact}, we define the notion of EXACT-distillation as follows.

\begin{definition}[{\bf $t$-Bounded EXACT-Distillation}]
Let $P,Q$ be two counting problems, and let $t:\mathbb{N}\rightarrow\mathbb{N}$. A {\em $t$-bounded EXACT-distillation} from $P$ into $Q$ is a pair $(\mathsf{reduce},\mathsf{lift})$ of two polynomial-time procedures such that:
\begin{itemize}
\item Given $n\in\mathbb{N}$, and $t(n)$ instances $x_1,x_2,\ldots,x_{t(n)}$ of $P$ with $|x_i|=n$ for all $i\in[t(n)]$, $\mathsf{reduce}$ outputs an instance $y$ of $Q$ with $|y|\leq t(n)\log n$.
\item Given $n,x_1,x_2,\ldots,x_{t(n)},y$ and $Q(y)$, $\mathsf{lift}$ outputs $Q(x_1),Q(x_2),\ldots,Q(x_{t(n)})$.
\end{itemize}
\end{definition}

Further, we will use the following proposition.

\begin{proposition}[Lemma 17.4 in \cite{fomin2019kernelization}]\label{prop:cover}
Let $X,Y$ be finite sets, $p\in\mathbb{N}$, and $\beta: X^p\rightarrow Y$.\footnote{Here, $X^p=X\times X\times\cdots\times X$ where $X$ appears $p$ times.} We say that $y\in Y$ {\em covers} $x\in X$ if there exist $x_1,x_2,\ldots,x_p\in X$ such that $x_i=x$ for some $i\in[p]$ and $\beta(x_1,x_2\ldots,x_p)=y$. Then, there exists $y\in Y$ that covers at least $|X|/|Y|^{1/p}$ elements in $X$.
\end{proposition}

For a $t$-bounded EXACT-distillation, we prove the following statement.

\begin{lemma}\label{lem:EXACTdistillation}
Let $P,Q$ be two counting problems such that there exists a $t$-bounded EXACT-distillation from $P$ into $Q$ for some polynomially bounded function $t$. Then, $P\in$NP/poly. In particular, if $P$ is \#P-hard, then \#P $\subseteq$ ``NP/poly''.
\end{lemma}

\begin{proof}
Let $(\mathsf{reduce},\mathsf{lift})$ be a $t$-bounded EXACT-distillation from $P$ into $Q$ for some polynomially bounded function $t$. For every $n\in\mathbb{N}$, let $X_n$ be the set of instances of $P$ of size $n$, let $\alpha(n)=t(n)\log n$, and let $Y_{\alpha(n)}$ be the set of instances of $Q$ of size at most $\alpha(n)$; because we suppose (w.l.o.g.) that our alphabet is binary, $|X_n|=2^n$ and $|Y_{\alpha(n)}|\leq \alpha(n)\cdot 2^{\alpha(n)} = (t(n)\log n)\cdot n^{t(n)} \leq t(n)^{t(n)}$. Here, we suppose (w.l.o.g.) that $t(n)$ is large enough compared to $n$.

By Proposition \ref{prop:cover}, there exists $y_1\in Y_{\alpha(n)}$ that covers at least $|X_n|/|Y_{\alpha(n)}|^{1/t(n)}\geq |X_n|/t(n)$ elements in $X_n$. Let $Z_1=\{x\in X_n: y_1$ covers $x\}$. By Proposition \ref{prop:cover} again,  there exists $y_2\in Y_{\alpha(n)}\setminus\{y_1\}$ that covers at least $|X_n\setminus Z_1|/(|Y_{\alpha(n)}|-1)^{1/t(n)}\geq |X_n\setminus Z_1|/t(n)$ elements in $X_n$. Let $Z_2=\{x\in X_n\setminus Z_1: y_2$ covers $x\}$. Generally, we let $y_i\in Y_{\alpha(n)}\setminus\{y_1,\ldots,y_{i-1}\}$ be an element that covers at least $|X_n\setminus(Z_1\cup\cdots\cup Z_{i-1})|/t(n)$ elements in $X_n\setminus(Z_1\cup\cdots\cup Z_{i-1})$, where $Z_i$ is the set of new elements covered by $y_i$.

\begin{claim}
For every $i\in[t(n)]$, $|Z_1\cup\cdots\cup Z_i|\geq\min(|X_n|,C_i|X_n|)$ for $C_i=\frac{t(n)-1}{t(n)}C_{i-1}+\frac{1}{t(n)}$ where $C_1=\frac{1}{t(n)}$.
\end{claim}

\begin{proof}
We use induction on $i$. Notice that we have already proved the base case ($i=1$). Now, we suppose that the claim is correct for $i-1$, and let us prove it for $i$. If $|Z_1\cup\cdots\cup Z_i|=|X_n|$, then we are done, and hence we next suppose that this is not the case. Then, by the inductive hypothesis, 
\[\begin{array}{ll}
|Z_1\cup\cdots\cup Z_i| &= |Z_1\cup\cdots\cup Z_{i-1}|+|Z_i| \geq |Z_1\cup\cdots\cup Z_{i-1}| + \displaystyle{\frac{|X_n|-|Z_1\cup\cdots\cup Z_{i-1}|}{t(n)}}\\
&\geq C_{i-1}|X_n| + \displaystyle{\frac{|X_n|-C_{i-1}|X_n|}{t(n)} = \left(\frac{t(n)-1}{t(n)}C_{i-1}+\frac{1}{t(n)}\right)\cdot|X_n|}.
\end{array}\]
This completes the proof of the claim.
\end{proof}

Now, observe that the above recurrence evaluates to
\[C_i = \displaystyle{\frac{1}{t(n)}\cdot\sum_{j=1}^{i-1}(\frac{t(n)-1}{t(n)})^j = 1-(1-\frac{1}{t(n)})^i}.\]
Specifically, the last equality can be verified by induction on $i$. Setting $i=2t(n)n$, we have that $C_i|X_n|\geq \displaystyle{\left(1-(1-\frac{1}{t(n)})^i\right)|X_n|\geq (1-\frac{1}{e^n})|X_n|}$. Because $|X_n|=2^n$, this means that all elements in $X_n$ are covered. 
We conclude that there exists $S_n\subseteq Y_{\alpha(n)}$ of size $t(n)^{\OO(1)}\leq n^{\OO(1)}$ that covers all elements in $X_n$.

Having $S_n$ at hand, we are ready to present an ``NP/poly'' algorithm $A$ for $P$. Given an instance $x$ of $P$ of size $n$, the advice used is the encoding of $\{(y,s): y\in S_n, s=Q(y)\}$. Observe that, since $|S_n|\leq n^{\OO(1)}$ and $Q$ is well-behaved, the size of the encoding is bounded polynomially in $n$. Using nondeterminism,  $A$ guesses a set of $t(n)$ strings of size $n$ each, $x_1,x_2\ldots,x_{t(n)}$, such that at least one of these strings is $x$. Then, $A$ calls $\mathsf{reduce}$ on $x_1,x_2\ldots,x_{t(n)}$, and obtains a string $y$. If there exists $s$ such that $(y,s)$ belongs to the advice, then observe that this $s$ is unique (being $Q(y)$), and $A$ returns  the output of $\mathsf{lift}$ on $s$. Otherwise, it returns ``Do Not Know''. 

From the construction of the advice, and the correctness of $(\mathsf{reduce},\mathsf{lift})$, it should be clear that, for every computation path of $A$, the output is either $P(x)$ or ``Do Not Know'', and that there exists a computation path of $A$ whose output is $P(x)$. Further, since $\mathsf{reduce}$ and $\mathsf{lift}$ are polynomial-time procedures, we have that $A$ runs in nondeterministic polynomial-time. This completes the proof.
\end{proof}

Having Lemma \ref{lem:EXACTdistillation} at hand, we are ready to prove Theorem \ref{thm:exact}.

\lowerThmEXACT*

\begin{proof}
Let $A$ and $R$ be the EXACT-cross-composition from $P$ into $Q$  and the corresponding equivalence relation, respectively, in the premise of the theorem. Targeting a contradiction, suppose that $Q$ has a polynomial compression $(\mathsf{reduce},\mathsf{lift})$ into some parameterized counting problem $W$.
Since $W$ is well-behaved, for any $n\in\mathbb{N}$, we can compute, in polynomial time, $N_n\in\mathbb{N}$ such that for every $(x,k)\in \Sigma^\star \times\mathbb{N}_0$ of size at most $n$, $W(x,k)\leq N_n$. We define a new parameterized counting problem, called EXACT$(W)$, as follows. An instance of EXACT$(W)$, $(z,k)$,  is of the form $z=z_1\#z_2\# . . . \#z_q$ and $k=\sum_{i=1}^qk_i$ for some $q\in\mathbb{N}$, where, for every $i\in[q]$, $(z_i,k_i)$ is an instance of $W$; then, $W(z,k)=\sum_{i=1}^qM_n^iW(z_i,k_i)$, where $M_n$ is determined later.  For some polynomially bounded function $t$, we will construct a $t$-bounded EXACT-distillation $(\mathsf{reduce}',\mathsf{lift}')$ from $P$ into EXACT$(W)$.  Due to Lemma \ref{lem:EXACTdistillation}, this will yield a contradiction, which will conclude the proof of the theorem.

Towards the construction of the above-mentioned EXACT-distillation, we identify three constants, $c_1,c_2$ and $c_3$:
\begin{enumerate}
\item Let $c_1$ be a fixed constant such that, given instances $x_1,x_2,\ldots,x_\ell$ of $P$ for some $\ell\in\mathbb{N}$ that are equivalent with respect to $R$, $A$ outputs an instance $(y,k)$ of $Q$  such that $k\leq(\max_{i=1}^\ell|x_i|+\log\ell)^{c_1}$. 
\item Let $c_2$ be a fixed constant such that, given an instance $(y,k)$ of $Q$, $\mathsf{reduce}$ outputs an instance $(z,k')$ of $W$ such that $|z|,k'\leq k^{c_2}$.
\item Let $c_3$ be a fixed constant such that, for every $n\in\mathbb{N}$, any set of strings that are each of size $n$ can be partitioned in polynomial time into at most $n^{c_3}$ $R$-equivalent classes.
\end{enumerate}

Let $t(n)=n^{2(c_1\cdot c_2+c_3)}$ be a polynomially bounded function. Let $M_n=n^{c_3}\cdot N_n$.

First, we describe the procedure $\mathsf{reduce}'$. For this purpose, let $x_1,x_2,\ldots,x_{t(n)}$ be $t(n)$ instances of $P$ of size $n$. We partition them, in polynomial time, into $R$-equivalent classes $X_1,X_2,\ldots,X_r$, where $r\leq n^{c_3}$. For every $i\in[r]$, we call $A$ on $X_i$ to obtain, in polynomial time, an instance $(y_i,k_i)$ of $Q$ with $k_i\leq (n+\log t(n))^{c_1}$. Then, for every $i\in[r]$, we call $\mathsf{reduce}$ on $(y_i,k_i)$ to obtain, in polynomial time, an instance $(z_i,k'_i)$ of $W$ with $|z_i|,k'_i\leq k_i^{c_2}\leq (n+\log t(n))^{c_1\cdot c_2}$.  The output instance of EXACT($W$) is $(z,k')$ where $z=z_1\#z_2\# . . . \#z_r$ and $k'=\sum_{i=1}^rk_i$. So, the running time is polynomial, and $|z|,k'\leq r+r\cdot (n+\log t(n))^{c_1\cdot c_2}\leq n^{2(c_1\cdot c_2+c_3)}=t(n)$. 

Second, we describe the procedure $\mathsf{lift}'$. Here, we are given $x_1,x_2,\ldots,x_{t(n)},(z,k')$ and $s'=\mathrm{EXACT}(W)(z,k')=\sum_{i=1}^qM_n^iW(z_i,k_i)$. By the choice of $N_n$, for every $i\in[r]$, we have that \[\begin{array}{ll}
M_n^i& \geq r\cdot N_n\cdot M_n^{i-1}\\
& >(i-1)\cdot N_n\cdot M_n^{i-1}\\
& \geq \sum_{j=1}^{i-1}N_n\cdot M_n^j\\
& \geq \sum_{j=1}^{i-1}W(z_j,k_j)\cdot M_n^j.
\end{array}\]
Hence, we can extract from $s'$, in polynomial time, $s'_1,s'_2,\ldots,s'_r$ such that $s'_i=W(z_i,k_i)$ for all $i\in[r]$. Indeed, this can be done by the following procedure:
\begin{enumerate}
\item Initialize $\widehat{s}\leftarrow s'$.
\item For $i=r,r-1,\ldots,1$:
	\begin{enumerate}
	\item Let $s'_i\leftarrow \lfloor\widehat{s}/M_n^i\rfloor$.
	\item Update $\widehat{s}\leftarrow \widehat{s}-s_i\cdot M_n^i$.
	\end{enumerate}
\item Return $s'_1,s'_2,\ldots,s'_\ell$.
\end{enumerate}
Then, using $\mathsf{lift}$ on each of $(y_i,k_i),(z_i,k'_i),s'_i$, $i\in[r]$, we obtain, in polynomial time, $s_1,s_2,\ldots,s_r$ such that $s_i=Q(y_i,k_i)$ for every $i\in[r]$. Finally, by the second item in the definition of an EXACT-cross-composition (Definition \ref{def:exactCross}), for each $i\in[r]$, we can apply a polynomial-time procedure that, given $X_i,(y_i,k_i)$ and $s_i$, outputs $P(x)$ for every $x\in X_i$. Thus, we derive $P(x_1),P(x_2),\ldots,P(x_{t(n)})$. This completes the proof.
\end{proof}

\end{document}